\documentclass[11pt,a4paper]{amsart}
\usepackage{amsaddr}
\usepackage{amsfonts,amsthm, amscd, epsfig, amsmath, amssymb,enumerate,bbm}

\usepackage{hyperref}
\hypersetup{
    colorlinks=true,
    linkcolor=blue,
    citecolor=red,
    pdfpagemode=FullScreen,
    pdftitle={Aletti,Crimaldi,Ghiglietti}
    }

\usepackage{graphicx}
\usepackage{url}
\usepackage{bbm}

\usepackage{tikz}
\foreach \a in {q,w,e,r,t,y,u,i,o,p,a,s,d,f,g,h,j,k,l,z,x,c,v,b,n,m,Q,W,E,R,T,Y,U,I,O,P,A,S,D,F,G,H,J,K,L,Z,X,C,V,B,N,M}{
  \expandafter\xdef\csname v\a\endcsname {
		{\noexpand\boldsymbol{\a}}
	}
}

\newcommand{\vone}{{\boldsymbol{1}}}
\newcommand{\vzero}{{\boldsymbol{0}}}

\theoremstyle{plain}
\newtheorem{theorem}{Theorem}[section]
\newtheorem{lemma}[theorem]{Lemma}

\theoremstyle{remark}

\theoremstyle{assumption}

\theoremstyle{definition}

\allowdisplaybreaks[3]

\author[G. Aletti]{Giacomo Aletti}
\address{ADAMSS Center,
  Universit\`a degli Studi di Milano, Milan, Italy}
\email{giacomo.aletti@unimi.it}

\author[I. Crimaldi]{Irene Crimaldi}
\address{IMT School for Advanced Studies Lucca, Lucca, Italy}
\email{irene.crimaldi@imtlucca.it}

\author[A. Ghiglietti]{Andrea Ghiglietti}
\address{Universit\`a degli Studi di Milano-Bicocca, Milan, Italy}
\email{andrea.ghiglietti@unimib.it}

\title{Interacting Innovation processes: \\case studies from Reddit and Gutenberg}
\thanks{All the authors contributed equally to the present work.}

\begin{document}

\begin{abstract}
In this work, we introduce an extremely general model for a collection of 
innovation processes in order to model and analyze the {\em interaction} among them. 
We provide theoretical results, analytically proven, and we show how the proposed model fits 
the behaviors observed in some real data sets (from {\em Reddit} and {\em Gutenberg}).  
It is worth mentioning that the given applications are only examples of 
the potentialities of the proposed model and related results: 
due to its abstractness and generality, it can be applied to many interacting innovation processes.
\end{abstract}

\maketitle

Analyzing the {\em innovation} process, that is the
underlying mechanisms through which novelties emerge, diffuse and
trigger further novelties is definitely of primary importance in
many areas (biology, linguistics, social science and others
\cite{armano, arthur, fink, gooday, obrien, puglisi, reader, rogers,
  rzh, saracco, sole, thurner}). We can define 
  {\em novelties } (or {\em innovations}) as the first time
occurrences of some event. A widely used mathematical object that models an 
innovation process is an {\em urn model with infinitely many colors},
also known as {\em species sampling sequence} \cite{han-pitman,
  pitman_1996, Zabell_1992}. Let $C_1$ be the first observed color, then,
given the colors $C_1,\dots, C_t$ of the first $t$ extractions, the
color of the $(t+1)$-th extracted ball is new 
(i.e. not already drawn in the previous extractions) 
with a probability
$Z^*_t$ which is a function of $C_1,\dots,C_t$ (sometimes called ``birth
probability'') and it is equal to the already observed color $c$ with
probability $P_{c,t}=\sum_{n=1}^t Q_{n,t}I_{\{C_n=c\}}$, where
$Q_{n,t}$ is a function of $C_1,\dots,C_t$. The quantities $Z^*_t$ and
$Q_{n,t}$ specify the model: precisely, $Z^*_t$ describes the
probability of having a new color (that is a novelty) at time-step
$t+1$ and $Q_{n,t}$ is the weight at time-step $t$ associated to extraction  
$n$, with $1\leq n\leq t$, so that the probability of having at
time-step $t+1$ the ``old'' color $c$ is proportional to the total
weight at time-step $t$ associated to that color (a {\em reinforcement
mechanism}, sometimes called ``weighted preferential attachment'' principle).
Note that the number of possible colors is not fixed a priori, but new
colors continuously enter the system. We can see the urn with
infinitely many colors as the space of possibilities, while the
sequence of extracted balls with their colors represents the history
which has been actually realized.  \\

\indent The {\em Blackwell-MacQueen urn} scheme \cite{black-mac,
  pitman_1996} provides the most famous example of innovation process.  
  According to this model, at time-step $t+1$ a new color is
observed with probability given by a deterministic function of $t$, 
that is $Z^*_t=z^*(t)=\theta/(\theta+t)$, where $\theta>0$,
and an old color is observed with a probability proportional to the
number $K_{c,t}$ of times that color was extracted in the previous
extractions: $Q_{n,t}=q_n(t)=1/(\theta+t)$,
i.e.~$P_{c,t}=K_{c,t}/(\theta+t)$.  This is the ``simple''
preferential attachment rule, also called ``popularity''
principle. This urn model is also known as Dirichlet process
\cite{ferguson} or as Hoppe's model~\cite{hoppe} and, in terms of
random partitions, it corresponds to the so called {\em Chinese
  restaurant process} \cite{pitman}. Afterwards, it has been extended
introducing an additional parameter and it has been called {\em
  Poisson-Dirichlet model}~\cite{james, pitman, pit-yor,
  Teh2006}. More precisely, for the Poisson-Dirichlet model, we have
\begin{equation}\label{PD-probab}
  \begin{split}
  &Z^*_t=\frac{\theta+\gamma D_t}{\theta +  t},
  \qquad Q_{n,t}=\frac{1-\gamma/K_{C_n,t}}{\theta + t},
  \\
&\mbox{and so }\quad
P_{c,t}=
\frac{K_{c,t}-\gamma}{\theta +  t},
\end{split}
\end{equation}
where $0\leq \gamma<1$, $\theta>-\gamma$ and $D_t$ denotes the number
of distinct extracted colors until time-step $t$.  
From an applicative point of view, as an innovation process,
the Poisson-Dirichlet process has the merit to reproduce in many cases
the correct basic statistics, namely the Heaps'~\cite{heaps_1978,
  herdan_1960} and the (generalized) Zipf's
laws~\cite{zipf_1929,zipf_1935,zipf_1949}, which quantify,
respectively, the~rate at which new elements appear and the frequency
distribution of the elements. In particular, the {\em Heaps' law} states 
that the number $D_t$ of distinct
observed elements (i.e.~colors, according to the metaphor of the urn)
when the system consists of $t$ elements (i.e.~after $t$ extractions
from the urn) follows a power law with an exponent smaller than or equal to $1$ and,  
 for the Poisson-Dirichlet model, we have
$D_t\propto t^{\gamma}$ for $0<\gamma<1$ (while  
$D_t\propto\ln(t)$ for $\gamma=0$).
\\

\indent Recently, a new model, called {\em urn with triggering},
that includes the Poisson-Dirichlet process as a particular case, 
have been introduced and studied \cite{ale-cri_SR, Tria3, Tria1, Tria2}.  
This model is based on Kauffman's principle of the adjacent
possible~\cite{kauffman_2000}: indeed, the model starts with an urn
with a finite number of balls with distinct colors and,
whenever a color is extracted for the first time, a set of balls with
new colors is added to the urn. This represents Kauffman's idea that,
when a novelty occurs, it triggers further potential novelties. 
In particular, the urn with triggering has the merit to provide  
a very clear representation of the evolution dynamics of the Poisson-Dirichlet process. 
An urn initially contains $N_0>0$ distinct balls of different colors. 
Then, at each time step $t+1$, a ball is drawn at random from the urn and 
\begin{itemize}
\item[(a)] if the color of the extracted ball is new, 
i.e. it was not been extracted in the previous extractions, 
then we replace the extracted ball by $\widehat{\rho}$ 
  balls of the same color as the extracted ball plus 
  $(\nu+1)$ balls of distinct new colors, i.e. not already present in the urn;
 \item[(b)] if the color of the extracted ball is old, i.e. it has
   been already extracted in the previous extractions, we replace the extracted 
   ball by $(1+\rho)$ balls of the same color as the extracted one.
 \end{itemize}
It easy to verify that, when the {\em balance condition}  
$\widehat{\rho}+\nu=\rho$ is satisfied (this means that at each time-step the number of balls added to the urn 
is always $\rho$, regardless of the outcome of the extraction), the above updating rule gives rise to 
the above probabilities \eqref{PD-probab}, taking $\rho>\nu\geq 0$, 
$\theta=N_0/\rho$ and $\gamma=\nu/\rho$.
\\

\indent 
{Since it is doubtless important to understand how different innovation processes affect each other,   
this work aims at introducing and analyzing a model for a finite network of innovation processes. In the proposed model, for each node $h$,  
 the probability of observing a new or an old item depends, not only on the path of observations 
 recorded for $h$ itself, but also on the outcomes registered for the other nodes $j\neq h$. More precisely, 
 we introduce a system of $N$ urns with triggering that interact each other as follows:}
 \begin{itemize}
 \item[(i)] 
{the probability of {\em exploitation} of an old item $c$ by node $h$, i.e.  
 the probability of extracting from urn $h$ a color $c$ already drawn in the past from an urn of the system 
 (not necessarily from $h$ itself),    
has an increasing dependence not only on the number of times $c$ has been observed in node $h$ itself 
(that could be even zero), but also on the number of times $c$ has been observed in each of the other nodes;}
\item[(ii)] 
{the probability of {\em production} (or {\em exploration}) of a novelty for the entire system 
by node $h$, i.e. the probability of extracting from urn $h$ a color never extracted before from any of the urns in the system,  
has an increasing dependence not only on the number of novelties produced by $h$ itself in the past, 
but also on the number of novelties produced by each of the other nodes in the past.}
\end{itemize}
In particular, (ii) means that Kauffman's principle of the adjacent possible is at the ``system level": that is, 
when urn $h$ produces a novelty for the system, this fact triggers further potential novelties 
in all the urns of the system, not only in urn $h$ itself.   The two different dependencies described above ((i) and (ii))  are 
tuned by two different matrices (called $\Gamma$ and $W$ in the sequel). 
\\

\indent Despite the amount of scientific works regarding interacting urns with a finite set of colors 
(see, for instance, \cite{ale-cri-ghi, ale-cri-ghi-complete} and the references therein),  
in the existing literature we have found only 
{a few papers about a collection of interacting 
(in the same sense of the present work) urns with infinitely many colors, 
that is \cite{fortini, iacopini-2020, ubaldi-2021}.  In the model provided in \cite{fortini} (see Example~3.8 in that paper),  
there is} a finite collection of 
Dirichlet processes with random reinforcement. More precisely, in that model we have a random weight 
$W_{t,h}$ associated to the extraction at time-step $t$ from the urn $h$ so that,  
the probability of extracting from urn $h$ an old  color $c$ 
(here, the term ``old" refers to urn $h$, that is a color never extracted before from urn $h$) 
is proportional to the weight associated to that color, specifically 
$\sum_{n=1}^t W_{n,h}I_{C_{n,h}=c}/ ( \theta +\sum_{n=1}^t W_{n,h})$. 
The interaction across the urns is introduced by means of the weights, which could be stochastically dependent:  
each $W_{n,h}$ may be the same for each urn $h$, or a function of the observed outcomes of the other urns, 
or a function of some common (observable or latent) variables. 
It is easy to understand that this model is different from ours: 
we consider Poisson-Dirichlet processes, not only Dirichlet processes, and, 
differently from the model in \cite{fortini}, for us, 
the notion of ``old" or ``new" color refers to the entire system, not to each single urn, and 
Kauffman's principle of the adjacent possible is at the system level as explained above.  
{Our work and \cite{iacopini-2020} share the fact that the proposed models are both a collection of
urns with triggering with an interacting dynamics that brings the Kauffman’s principle of the
adjacent possible from the single agent to the network of agents. Adopting the terminology of 
\cite{iacopini-2020}, we can say that both interacting mechanisms are based on the construction and
the updating of a “social" urn for each network node from which the extractions take place, but the
contruction and the updating rules of the social urns are deeply different in the two models. In
particular, differently from \cite{iacopini-2020}, we introduce the notion of "new” and "old” at the
system level. In \cite{iacopini-2020} the authors focus on the novelties in each sequence (novelty in $h$
= first apperance in $h$ of a new item), that they call “discoveries”; while we also study the
sequence of the novelties for whole the system produced by each agent. Furthermore, in \cite{iacopini-2020} 
the extraction of an “old” item in a certain network node does not affect the
other nodes, in our model we also have an interacting reinforcement mechanism for the “old” items:
the probability of the extraction of an old item depends on the number of times it has been observed
in all the nodes. This allows us to get a specific result on the distribution of the observations in the
system among the different items observed. Finally, \cite{ubaldi-2021} provides a multi-agent version of 
the urn with triggering model, which is specific for describing the birth and the evolution of social networks.}
\\

\indent While the model we propose is extremely general and may be also employed in other contexts, 
it has been tested on two real data sets: 
one taken from the social content aggregation website {\em Reddit}, 
collected, elaborated and made freely available on the web by the authors of  
\cite{monti}, and one got from the on-line library {\em Project Gutenberg}, which is a collection of public domain books.  
{We show that both data sets exhibit empirical behaviours that are 
in accordance with those predicted by the proven theoretical results.}
\\

\indent The sequel of the paper is so structured. In Section~\ref{sec-methods} we 
introduce the model and we explain the role played by each model parameter. 
In Section \ref{sec-results}  we illustrate the theoretical results and we show how some real innovation processes can be 
well described using the proposed model. Section \ref{discussion} is devoted to the discussion of the achieved results and 
the presentation of possible future developments. Finally, the supplementary material collects the 
analytical proof of all the presented theoretical results. 


\section{Methods}\label{sec-methods}

The model we propose essentially consists in 
a finite system of {\em interacting} urns with triggering. More precisely,  
suppose to have $N$ urns (that may represent $N$ different agents of a system), 
labeled from $1$ to $N$. 
{{\em At time-step $0$, the colors inside each urn are different from those in the other urns}. 
Let $N_{0,h}>0$ be the number of distinct balls with distinct colors inside the urn $h$.}  
Then, at each time-step $t\geq 1$, one ball is drawn at
random from each urn and, for any $h=1,\dots N$, urn $h$ is so updated according to the colors extracted
 from urn $h$ itself and from all the other urns $j\neq h$:
\begin{itemize}
\item if the color of the ball extracted from urn $h$ is ``new'' (i.e., it appears
  for the first time in the system), then we replace (inside urn $h$) the extracted ball by 
  $\widehat{\rho}_{h,h}>0$ balls of the same color plus 
{$(\nu_{h,h}+1)$, with $\nu_{h,h}\geq 0$,} 
  balls of distinct ``new'' colors (i.e. not already present in the system);
 
 \item if the color of the ball extracted from urn $h$ is ``old'' (i.e., it has
   been already extracted in the system), we add $\rho_{h,h}>0$ balls of the same color 
   into urn $h$; 

\item for each $j\neq h$, if the color of the ball extracted from urn $j$ is ``new'' (i.e., it appears
  for the first time in the system), then into urn $h$ we add $\widehat{\rho}_{j,h}\geq 0$ balls of 
  the same color as the 
{one} extracted from urn $j$ 
  plus $\nu_{j,h}\geq 0$ balls of distinct ``new'' colors (i.e. not already present in the system);
  
 \item for each $j\neq h$, if the color of the ball extracted from urn $j$ is ``old'' (i.e., it has
   been already extracted in the system), then into urn $h$ we add $\rho_{j,h}\geq 0$ 
   balls of the same color as the one extracted from urn $j$.
 \end{itemize}
As already pointed out, the terms ``new" and ``old" refer to the entire system, 
that is a ``new" color is a color that has never been extracted from an urn of the system. 
On the contrary,  an ``old" color is a color that has already been extracted from at least one urn of the system, 
but it is possible that it has never been extracted from some other urns in the system.
\\
\indent We assume that the {\em ``new" colors added to a certain urn are 
always different from those added to the other urns 
(at the same time-step or in the past)}.
By means of this fact, together with the assumption that initially the colors in the urns are different from each other, 
we cannot have the same new color extracted simultaneously from different urns. 
In other words, {\em we cannot have the same novelty produced simultaneously from different agents 
of the system}. 
{Therefore, for each observed new color (novelty) $c$, there exists a unique urn (agent), say $j^*(c)$, 
in the system that produced it. However, in a time-step following its first extraction,  
color $c$ could be also extracted from another urn $h\neq j^*(c)$,      
as a consequence of the interaction among the urns (agents). Indeed, the ``contamination" of the color-set of the urn $h$ 
with the colors present in the other urns is possible by means of the interaction terms 
$\widehat{\rho}_{j,h}$ and/or $\rho_{j,h}$ in the above model dynamics.}
\\

\indent As in the standard Poisson-Dirichlet model, we assume the {\em balance condition} 
\begin{equation}\label{balance-cond}
\widehat{\rho}_{j,h}+\nu_{j,h}=\rho_{j,h},\quad \mbox{i.e. } \widehat{\rho}_{j,h}=\rho_{j,h}-\nu_{j,h}\,,
\end{equation} 
so that, at each time-step, each urn $j$ contributes to increase the number of balls inside urn $h$ by $\rho_{j,h}\geq 0$, 
with $\rho_{h,h}>0$. Therefore, 
at each time-step, the number of balls added to urn $h$ is $\rho_h=\sum_{j=1}^N \rho_{j,h}>0$.  
Hence,  if we denote by $C_{t+1,h}$ the color extracted from urn $h$ at time-step $t+1$, we have 
\begin{equation*}
Z^*_{t,h}=P(C_{t+1,h}=\mbox{``new"}\,|\, \mbox{past})=
\frac{N_{0,h}+\sum_{j=1}^N \nu_{j,h}D^*_{t,j}}{N_{0,h} + \rho_h t}\,,
\end{equation*}
where $D^*_{t,j}$ denotes the number, until time-step $t$, 
of distinct observed colors extracted for their first time from urn $j$, 
that is the number of distinct novelties for the whole system 
``produced" by urn (agent) $j$ until time-step $t$.
Moreover, for each old color $c$, we have 
\begin{equation*}
\begin{split}
P_{t}(h,c)=P(C_{t+1,h}= c \,|\, \mbox{past} ) &= 
\frac{\sum_{j\neq j^*(c)}\rho_{j,h}K_t(j,c)+\rho_{j^*(c),h}(K_t(j^*(c),c)-1)+\widehat{\rho}_{j^*(c),h}}{N_{0,h} + \rho_h t}
\\
&=\frac{\sum_{j=1}^N\rho_{j,h}K_t(j,c)+(\widehat{\rho}_{j^*(c),h}-\rho_{j^*(c),h})}{N_{0,h} + \rho_h t}
\\
&=\frac{\sum_{j=1}^N\rho_{j,h}K_t(j,c)-\nu_{j^*(c),h}}{N_{0,h} +\rho_h t}\,,
\end{split}
\end{equation*}
where $K_{t}(j,c)$ denotes the number of times the color $c$ has been extracted from urn $j$ until time-step $t$ and 
$j^*(c)$ denotes the urn from which the color $c$ has been extracted for the first time. 
(Note that $\rho_{j^*(c),h}=0$ implies $\nu_{j^*(c),h}=0$ by the balance condition.)
\\

\indent Without loss of generality,
to ease the notation we adopt a different parametrization by setting
\begin{equation}\label{normalization}
\theta_h=N_{0,h}/\rho_h,\quad \gamma_{j,h}=\nu_{j,h}/\rho_h,\quad  
\lambda_{j,h}=\widehat{\rho}_{j,h}/\rho_h\quad\mbox{and}\quad 
w_{j,h}=\rho_{j,h}/\rho_h\,,
\end{equation}
where $\theta_h>0$, $0\leq\gamma_{j,h}\leq 1$ with $<1$ for $j=h$, 
$0\leq \lambda_{j,h}\leq 1$ with $>0$ for $j=h$ and 
$0\leq w_{j,h}\leq 1$ with $>0$ for $j=h$.
This choice can be read as a normalization of the parameters since, for each $h=1,\dots, N$, we have
$\sum_j w_{j,h}=1$ and so, by the balance condition, $0\leq \sum_j\gamma_{j,h}<1$ and $0<\sum_j\lambda_{j,h}\leq 1$. 
With the new parametrization, we obtain
\begin{equation}\label{birth-prob-inter}
Z^*_{t,h}=P(C_{t+1,h}=\mbox{``new"}\,|\, \mbox{past})=
\frac{\theta_{h}+\sum_{j=1}^N \gamma_{j,h}D^*_{t,j}}{\theta_{h} + t}\,,
\end{equation}
and, for each "old" color $c$, 
\begin{equation}\label{old-color-prob-inter}
\begin{split}
P_{t}(h,c)=P(C_{t+1,h}= c \,|\, \mbox{past} ) &= 
\frac{\sum_{j\neq j^*(c)}w_{j,h}K_t(j,c)+w_{j^*(c),h}(K_t(j^*(c),c)-1)+\lambda_{j^*(c),h}}{\theta_h + t}
\\
&=\frac{\sum_{j=1}^N w_{j,h}K_t(j,c)-\gamma_{j^*(c),h}}{\theta_h + t}\,.
\end{split}
\end{equation}

Note that the probability that urn (agent) $h$ will produce at time-step $t+1$ a novelty for the entire system has an 
increasing dependence on the number $D^*_{t,j}$ of novelties produced by the urn (agent) $j$ until  
time-step $t$ and the parameter $\gamma_{j,h}$ regulates this dependence. In other words,  
Kauffman's principle of the adjacent possible is at the ``system level": that is,  for each pair $(j,h)$ of urns in the system, 
the parameter $\gamma_{j,h}$ quantifies how much the production of a novelty by urn $j$ induces potential novelties in urn $h$. 
Moreover, on the other hand, 
the probability that from urn $h$ we will extract at time-step $t+1$ an old color $c$ has an increasing dependence 
on the number $K_{t}(j,c)$ of times the color $c$ has been drawn from urn $j$ until time-step $t$ and the parameter 
$w_{j,h}$ quantifies how much the number $K_t(j,c)$ leads toward a future extraction of a ball of color $c$ from urn $h$. 
\\

\indent As particular cases, we can see that the case $N=1$ reduces to the classical Poisson-Dirichlet process with 
parameters $\theta>0$ and $0\leq \gamma<1$, and the {\em case of independence}  
corresponds to the framework when $w_{j,h}=0$ (and so $\gamma_{j,h}=\lambda_{j,h}=0$) 
for each $j\neq h$. In the latter case, by the model definition, the colors are not shared by the agents, because 
 each urn has colors different from those inside the other urns.  
Indeed, for $N$ independent Poisson-Dirichlet processes the probability of having colors in common is null.

\subsubsection*{Chinese restaurant metaphor} 

It is also worthwhile to recall that a standard metaphor used 
to represent the random partition induced by the Poisson-Dirichlet process, that is the 
random partition of the extracted balls among the observed colors, is the 
``Chinese restaurant" metaphor: suppose to have a restaurant with infinite tables 
and, at each time-step $t+1$, a customer enters and sits at a table,   
with probabilities $Z^*_t$ and $P_{c,t}$ given in \eqref{PD-probab} 
as the probability of sitting to an empty table and to an already occupied table, respectively. 
The random partition induced at time-step $t$ is the random allocation of the customers, arrived until time-step $t$, 
among the occupied tables.  
The interacting model introduced above can be represented with a similar metaphor. More precisely, 
suppose to have a restaurant with infinite tables where, at each time-step, $N$ customers enter simultaneously. 
Each customer belongs to a specific category $h=1,\dots, N$.
Then, at time-step $t+1$, the probability that the customer belonging to category $h$ sits to an empty table 
is $Z^*_{t,h}$ defined in \eqref{birth-prob-inter} and the probability that she sits to an already occupied table is 
$P_t(h,c)$ defined in \eqref{old-color-prob-inter}. We cannot have customers belonging to different categories that occupy simultaneously 
the same empty table. 
{ However, the sharing of a table by multiple categories is possible, 
after the first occupation of the table,  because of the 
presence of the interaction terms $\lambda_{j,h}$ and $w_{j,h}$ in \eqref{old-color-prob-inter}.}   
The probability $Z^*_{t,h}$ results increasing not only with the number of distinct 
tables occupied by customers of category $h$ until time-step $t$, but also with the 
numbers of distinct tables occupied by customers of each other category $j\neq h$. 
The parameters $\gamma_{j,h}$ rule these dependencies.
Similarly, the probability $P_t(h,c)$ has naturally an increasing dependence on the number of customers already seated at that table, 
but each of these customers has a different weight, i.e. $w_{j,h}$, according to her category: indeed, 
the parameter $w_{j,h}$ regulates how much the number of customers of category $j$ sitting to a table drives a customer of category $h$ to 
choose that table. 
{For the sake of clarity, we have synthetize in Table \ref{tab:table-correspondence-1}  
how the quantities and the events involved in the proposed model can be interpreted 
through both the urn metaphor or the Chinese restaurant metaphor.} 
\\

\begin{table}[h!]
{
  \begin{center}
    \caption{Correspondence table between the model, the urn metaphor and the Chinese restaurant metaphor}
    \label{tab:table-correspondence-1}
\begin{tiny}
    \begin{tabular}{l|l|l} 
    \hline
      \textbf{} & \textbf{Urn metaphor} & \textbf{Chinese restaurant metaphor}\\
      \hline
      agents & urns & categories\\
      \hline
      agent's action & extracted ball & customer entering the restaurant\\
      \hline
      item adopted & color of the extracted ball & table chosen by the customer\\
      \hline
      production (or exploration) of a novelty & extraction of a color never extracted before  & occupation of an empy table\\
      & from any urn of the system & \\
      \hline
      exploitation of an old item & extraction of a color already extracted before  & choice of an already occupied table\\
      & from some urn in the system & \\
      \hline
      $D^*_{t,h}=$ number, until time-step $t$, & number, until time-step $t$, of distinct colors
      & number, until time-step $t$, \\
      of distinct novelties for the whole system &  observed in the whole system and & of distinct tables\\
      produced by agent $h$ & extracted for their first time from urn $h$ & occupied for their first time by \\
      & & a customer belonging to category $h$\\
      \hline 
      $D_{t,h}=$ number, until time-step $t$, &  number, until time-step $t$, & number, until time-step $t$,  \\
      of distinct items adopted by agent $h$ & of distinct colors extracted from urn $h$ & of distinct tables occupied by at least one\\
      & &  customer belonging to category $h$\\
      \hline
      $K_t(h,c)=$ number, until time-step $t$, & number, until time-step $t$, & number, until time-step $t$, of customers\\
      of times agent $h$ has adopted item $c$ & of times color $c$ has been extracted from urn $h$ &  belonging to category $h$ and 
      sitting at table $c$\\ 
      \hline
    \end{tabular}
\end{tiny}
  \end{center}
  }
\end{table}


\subsubsection*{Matrix notation}

In order to present the theoretical results, we set $\Gamma$, $W$, $\Lambda$ equal to the non-negative $N\times N$ square matrices 
with elements $\gamma_{j,h}$, $w_{j,h}$ and 
$\lambda_{j,h}$, respectively. We recall that, by the balance condition~\eqref{balance-cond}  
and the reparametrization \eqref{normalization}, we have 
$$
W=\Gamma+\Lambda,\qquad \mathbf{0}^\top\leq {\vone}^\top\Gamma<\vone^\top
\qquad \mbox{and}\qquad {\vone}^\top W={\vone}^\top,
$$
where ${\vone}$ and ${\vzero}$ denote the vectors with all the components equal to $1$ and $0$, respectively. 
As observed above, the matrix $\Gamma$ rules the production of potential novelties and, in particular, 
its elements out of the diagonal regulate the interaction among the agents with respect to this issue; while, 
the matrix $W$ rules the interaction among the agents with respect to the choice of an old item. 
\\

\section{Results}\label{sec-results}
In this section we will present first the theoretical results and then the empirical results related to two
real data sets. The proofs of the first ones are collected in the supplementary materials, that may be found together with the online version 
at \cite{ale-cri-ghi-supplM}.

\subsection{Theoretical results}\label{sec-model}

The first result states that, if $\Gamma$ is irreducible, that is the graph with the 
{agents} as nodes 
and with $\Gamma$ as the adjacency matrix  is strongly connected, then 
$D^*_{t,h}\propto t^{\gamma^*}$ a.s. for all $h=1,\dots,N$, that is 
all the $D^*_{t,h}$ grow with the same Heaps' exponent $\gamma^*\in (0,1)$. 
{This means that, at the steady state, all the agents of the network produce innovations 
for the system at the same rate.} In addition, 
the ratio $D^*_{t,h}/D^*_{t,j}$ provides a strongly consistent estimator of 
the ratio $u_h/u_j$ of the relative centrality scores (with respect to $\Gamma^\top$) 
of the two nodes $h$ and $j$. More precisely, we have

\begin{theorem}\label{th-synchro-rates}
Suppose that the matrix $\Gamma$ is irreducible. Denote by $\gamma^*\in (0,1)$ the 
Perron-Frobenius eigenvalue of $\Gamma$, by ${\vv}$ the corresponding right eigenvector with
strictly positive entries and such that ${\vv}^\top{\vone}=1$ and, finally,  denote 
by ${\vu}$ the corresponding left 
{eigenvector} with strictly positive entries and  
${\vv}^\top{\vu}=1$. 
Then, for each $h=1,\dots, N$, we have 
$$
t^{-\gamma^*}D^*_{t,h}\stackrel{a.s.}\longrightarrow D^{**}_{\infty,h}\,,
$$ 
where $D^{**}_{\infty,h}$ is a finite strictly positive random variable. 
Moreover, for each pair of indexes $h,j=1,\dots, N$, we have 
$$
\frac{ D^{*}_{t,h} }{ D^{*}_{t,j} }\stackrel{a.s.}\longrightarrow \frac{u_h}{u_j}\,.
$$
\end{theorem}
As a consequence, since the number $D_t^*$ of distinct 
{items}  
observed in the entire system until time-step $t$ coincides 
by model definition with $\sum_{h=1}^N D^*_{t,h}$, we also have 
{that this number grows as $t^{\gamma^*}$, i.e.}    
$$
t^{-\gamma^*}D^*_t\stackrel{a.s.}\longrightarrow D_\infty^{**}=\sum_{h=1}^N D^{**}_{\infty,h}\,.
$$ 
Furthermore, 
{ 
if we denote by $(D_{t,h})$ the {\em discovery} process \cite{iacopini-2020} for agent $h$, that is 
if we denote by $D_{t,h}$ the number of distinct items adopted by agent $h$, }   
then we have 
$D^*_{t,h}\leq D_{t,h}\leq D^*_t$ and so we get 
$$
D_{t,h}=O(t^{\gamma^*})\qquad\mbox{and}\qquad 1/D_{t,h}=O(t^{-\gamma^*}),
$$
which, in particular, imply that, when the quantities $D_{t,h}$ have an asymptotic power law behavior,  
then they necessarily have the same Heaps' exponents, equal to $\gamma^*$. In addition, we obtain 
$$
\frac{u_h}{\sum_{h=1}^N u_h}\leq\liminf_{t}\frac{D_{t,h}}{D_{t,j}}\leq \limsup_t\frac{D_{t,h}}{D_{t,j}}
\leq \frac{\sum_{h=1}^N u_h}{u_j}\,.
$$

\indent 
{The second result of the present work affirms that if $W$ is irreducible, 
that is the graph with the agents as nodes and $W$ as the adjacency matrix is strongly connected, 
then, for each observed item $c$, 
the number of times item $c$ has been adopted by agent $h$ grows linearly. Moreover, 
at the steady state, the times item $c$ has been adopted in the whole system are uniformly  
distributed among the agents. This concept can be reformulated more clearly using the metaphor of 
the Chinese restaurant: the limit composition of each table $c$ is the uniform one
(with respect to the categories).} 
More precisely, we have  

\begin{theorem}\label{th-uniform-partition}
Suppose that the matrix $W$ is irreducible.
Then, for each $h=1,\dots,N$, we have 
$$
\frac{1}{t}K_t(h,c)\stackrel{a.s.}\longrightarrow K_{\infty}(c)
$$
for each observed color $c$ in the system,
where $K_{\infty}(c)$ is a suitable random variable that takes values in $(0,1]$ and does not depend on $h$. 
As a consequence, for each $h=1,\dots,N$, we also have that
$$
\frac{K_t(h,c)}{\sum_{j=1}^N K_t(j,c)}\stackrel{a.s.}\longrightarrow \frac{1}{N}.
$$ 
\end{theorem}


\subsection{Empirical results}\label{sec:real data}

{In this subsection we show that the behaviors predicted by the previous theoretical
results match with the ones we actually observe in two different real data sets: } 
 one taken from the social content aggregation website {\em Reddit}, 
collected, elaborated and made freely available on the web by the authors of  
\cite{monti} at \href{https://github.com/corradomonti/demographic-homophily}{https://github.com/corradomonti/demographic-homophily}, 
and one got from the on-line library {\em Project Gutenberg} at \href{https://www.gutenberg.org/}{https://www.gutenberg.org/}.  
{In order to illustrate these examples, we adopt the metaphor of the Chinese restaurant and so, for each of them, we 
identify the customers' categories and the tables we are looking at. In both examples,  
we consider $N=2$ categories with their sequences of customers who select the tables.   
See Table \ref{tab:table-correspondence-2} for a guide on how to interpret the quantities and the events of interest 
in the considered data sets in terms of the Chinese restaurant metaphor.} 
\\

{\indent We analyze the processes $(D^*_{t,h})$ and $(D_{t,h})$, with $h=1,\,2$, and the composition of the tables,
constructed starting from the real data, in order to verify if 
they exhibit a behavior along time in agreement with the theoretical results of the previous section. 
Specifically, we point out:
\begin{itemize}
\item[1)] the power law behavior of the processes $(D^*_{t,h})$ and $(D_{t,h})$, with $h=1,2$;
\item[2)] the fact that the above processes increases with the same Heaps' exponent 
(the constant $\gamma^*$ in Theorem~\ref{th-synchro-rates});
\item[3)] the convergence of the ratio $D^*_{t,1}/D^*_{t,2}$, or equivalently of the difference 
$\log_{10}(D^*_{t,1})-\log_{10}(D^*_{t,2})$, as $t\to +\infty$;
\item[4)] the convergence of the composition of the tables toward the uniform one (as stated in Theorem~\ref{th-uniform-partition}).
\end{itemize}
For points 1) and 2), we follow the standard method in literature: 
we provide the $\log_{10}-\log_{10}$ plot of the considered processes and the estimate of the common slopes 
of the corresponding lines by a least square interpolation. The goodness of fit of the provided lines with the same slope is 
supported by the extremely high value of the $R^2$ index. 
Regarding point 3), we plot the observed sequence $\log_{10}(D^*_{t,1})-\log_{10}(D^*_{t,2})$ along time, 
in order to highlight how its fluctuations decrease along time 
and how it asymptotically stabilizes. The limit of this process is estimated as  the difference between the 
intercepts of the two lines obtained for $D^*_{t,1}$ and $D^*_{t,2}$ in the $\log_{10}-\log_{10}$ plot. 
This value, denoted as $\widehat{u}$, represents an estimation of the difference 
$\log_{10}(u_1/u_2)=\log_{10}(u_1)-\log_{10}(u_2)$, 
where $r=u_1/u_2$ is the limit quantity in the second part of Theorem~\ref{th-synchro-rates}, which is also 
the ratio of the two centrality scores with respect to $\Gamma^\top$ of the two categories. 
Finally, for point 4), we plot the quantiles of the distribution of the proportion 
$\frac{K_t(1,c)}{K_t(1,c) + K_t(2,c)}$, from the least populated table $c$ to the most populated one, 
in order to appreciate their convergence toward $1/2$.} 
\\
 
\begin{table}[h!]
{
  \begin{center}
    \caption{Correspondence table for applications}
    \label{tab:table-correspondence-2}
\begin{tiny}
    \begin{tabular}{l|l|l} 
    \hline
      \textbf{Chinese restaurant metaphor} & \textbf{Reddit data set} & \textbf{Gutenberg data set}\\
      \hline
      categories & sentiment types (positive/negative) & literary genres (Western/History)\\
      \hline
      customer entering the restaurant & comment & word\\
      \hline
      chosen table & author of the commented news & word\\
      \hline
      occupation of an empty table & first comment to a news & first appearance of a word \\
      \hline
      choice of an already occupied table & comment to a news whose author has already & usage of a word already used before\\
      & received comments to other her news & in some of the literary genres\\
      \hline
      $D^*_{t,h}=$ number, until time-step $t$, & number, until time-step t, & number, until time-step $t$,\\
       of distinct tables occupied & of distinct authors whose & of distinct words whose\\
       for their {\em first} time by & {\em first} received comment & {\em first} apperance has been \\
       a customer belonging to category $h$ & belongs to sentiment category $h$ & in literary genre $h$\\
      \hline
      $D_{t,h}=$ number, until time-step $t$, & number, until time-step $t$, & number, until time-step $t$,\\
       of distinct tables occupied by & of distinct authors who have received  & of distinct words that\\
       {\em at least} one customer  & {\em at least} one comment & have been used in\\
       belonging to category $h$ & belonging to sentiment category $h$ & literary genre $h$\\
      \hline
      $K_t(h,c)=$ number, until time-step $t$, of customers & number, until time-step $t$ of comments & number, until time step $t$, of times\\
      belonging to category $h$ and & received by author $c$ & the word $c$ has been used\\ 
      sitting at table $c$ & belonging to sentiment category $h$  & in literary genre $h$\\ 
      \hline
    \end{tabular}
\end{tiny}
  \end{center}
  }
\end{table}


\subsubsection*{Reddit data set}

This data set  consists of a collection of news, and comments associated to each news, 
for the period $2016-2020$, downloaded from the {\em r/news} community on 
the website {\em Reddit} at \href{https://www.reddit.com/r/news}{https://www.reddit.com/r/news}, which is devoted to 
the discussion of news articles about events 
in the United States and the rest of the world. 
Each news is associated with the author who posted it.   
Moreover, the data set contains the specific topic the news belongs to (we refer to \cite{monti} for details 
about the topic classification) and, to each comment 
is also assigned a measurement of the sentiment, expressed as a real value in $(-1,1)$. It corresponds to 
the ``compound" score given by the VADER (Valence Aware Dictionary and sEntiment Reasoner) Sentiment Analysis
\cite{vader}, which is a lexicon and rule-based sentiment analysis tool, 
specifically thought for sentiments expressed in social media. 
\\

\indent Here we consider only the comments to news belonging to the topic ``Politics".
{Moreover, we categorize the sentiment variable, following \cite{ale-cri-sar-sentiment}: precisely, we define it as }   
``positive" if the provided sentiment value was larger than $+0.35$ and ``negative" if the provided sentiment value was lower than $-0.35$. 
Any comment with an original sentiment value that lies within $-0.35$ and $+0.35$ has been removed.
Summing up, we consider all the comments to the commented news regarding the topic ``Politics", 
with a sentiment value larger than $+0.35$ (positive) or lower than $-0.35$ (negative).
This provides us a total of $3\, 016\, 990$ comments in the negative sentiment category and $2\,602\,173$ comments in the positive sentiment category.
\\

\indent 
{We are interested in the sequence of authors who receive at least one comment with negative sentiment 
for the news they post and in the analogous sequence related to comments with positive sentiment. As explained above, we illustrate  
that these two sequences exhibit the asymptotic behaviors predicted by the proposed model. For this purpose, }  
we firstly identify the main quantities related to the Chinese restaurant version of the model: 
each sentiment category is a customer category (category 1 = negative sentiment and category 2 = positive sentiment)  
and the authors represent the tables. 
Therefore, when at time-step $t$ a news receives a comment with a specific sentiment, then the author who posted such a news 
is ``new" or ``old" for that specific sentiment category if, 
respectively, she has or has not already received a comment within that sentiment category.
Analogously, the author will be ``new" or ``old" for the entire system (the whole collection of comments) if, respectively, 
she has or has not already received any comment. In order to obtain two sequences of comments of the same length,  as required by the model, 
we have randomly removed some comments from the negative sentiment category, i.e. the one containing more comments. In addition, 
we verified successfully that an author is not commented for the first time simultaneously with two comments of different sentiment. 
\\

\indent 
{For each possible sentiment category $h=1,\,2$, 
the observed quantity $D^*_{t,h}$ (i.e. the number, until time-step $t$, of 
distinct authors whose {\em first} received comment belongs to sentiment category $h$) shows a power law growth along time.  
We can observe the same behavior for $D_{t,h}$ (i.e. the number, until time-step $t$, of distinct authors who 
have received {\em at least one} comment belonging to sentiment category $h$).  
Figure \ref{fig:lines_Reddit-new} provides the asymptotic behavior of these processes in $\log_{10}-\log_{10}$ scale,  
where we can also appreciate how the lines exhibit the same slope, which indicates that the processes have the same Heaps' exponent.   
This is exactly in accordance with the first result of Theorem~\ref{th-synchro-rates}.  
The estimated value of the Heaps' exponent, estimated as the common slope of the lines in the $\log_{10}-\log_{10}$ plot  
is $\widehat{\gamma}^*=0.781$.  
Figure \ref{fig:ratio_Reddit} shows the convergence of the process $\log_{10}(D^*_{t,1})-\log_{10}(D^*_{t,2})$ 
toward the estimated limit value  $\widehat{u}=-0.727$, computed as the difference between the intercepts of the two regression lines 
 for the two processes $(D^*_{t,h})$ in Figure \ref{fig:lines_Reddit-new}. This value is an estimation of the 
quantity $u=\log_{10}(r)$, where $r=u_1/u_2=10^u$ is the limit in the second result of Theorem~\ref{th-synchro-rates}.}
\\

 \begin{figure}[htp]
\centering
\includegraphics[height=0.3\textheight, width=0.8\textwidth]{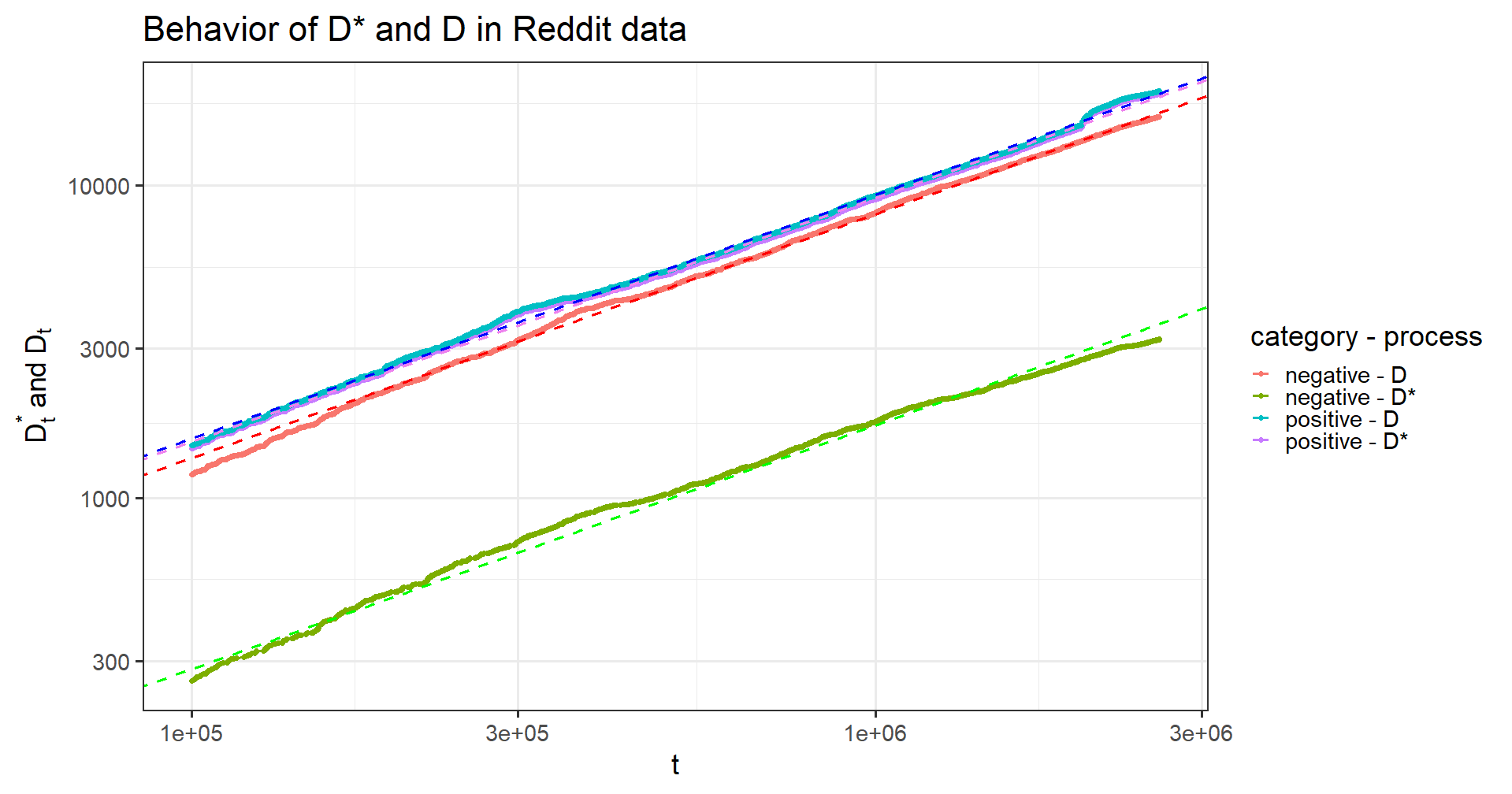}  
\caption[]{
{ Reddit data set.  
Linear behavior of $(D^{*}_{t,h})$ and $(D_{t,h})$ along time, for $h=1,2$, in $\log_{10}-\log_{10}$ scale. 
The dashed lines are obtained by a least square interpolation. The goodness of fit $R^2$ index is $0.9984$. 
The estimated common slop is $\widehat{\gamma}^*=0.781$.
}}
\label{fig:lines_Reddit-new}.  
\end{figure}

\begin{figure}[htp]
\centering
\includegraphics[height=0.3\textheight, width=0.8\textwidth]{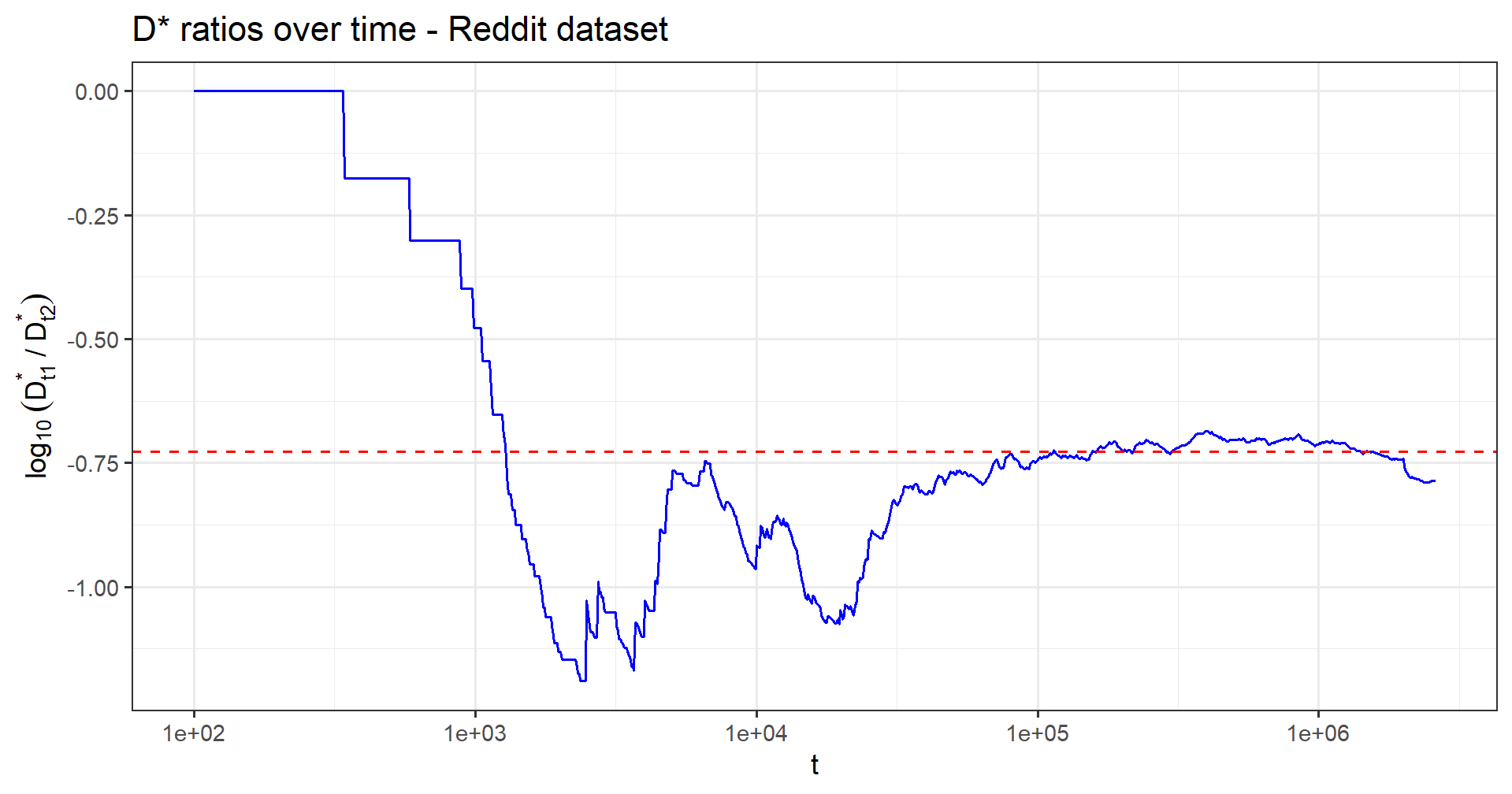}  
\caption[]{ 
{Reddit data set.  
Plot of the process $\log_{10}(D^*_{t,1})-\log_{10}(D^*_{t,2})$ along time. The horizontal dashed red line  
represents the estimated limit value $\widehat{u}=-0.727$.} }
\label{fig:ratio_Reddit}.  
\end{figure}

\indent Regarding the table composition, we provide Figure \ref{fig:Table_composition_Reddit-new} with  
the proportion of comments with negative sentiment received by an author over the total number of received comments. 
More precisely, we plot the quantiles of the empirical distribution of this proportion, 
from the least commented author to the most commented one.    
{In order to construct the quantiles, we have listed the authors (tables) from the least commented to the most commented, 
removing those commented less than $10$ times (tables with less than $10$ customers),  
then we have grouped these authors taking intervals of equal length ($0.5$ in $\log_{10}$ scale). 
Finally, within each group, we have computed the quantiles of the empirical distribution of  
the proportion of the comments with negative sentiment the authors have received with respect to the total number of received comments.}  
We can appreciate how these quantiles get closer to $1/2$ (the uniform composition) as the number of received comments increases. 
This is in accordance with Theorem~\ref{th-uniform-partition}. 
\\
 
\begin{figure}[htp]
\centering
\includegraphics[height=0.3\textheight, width=0.8\textwidth]{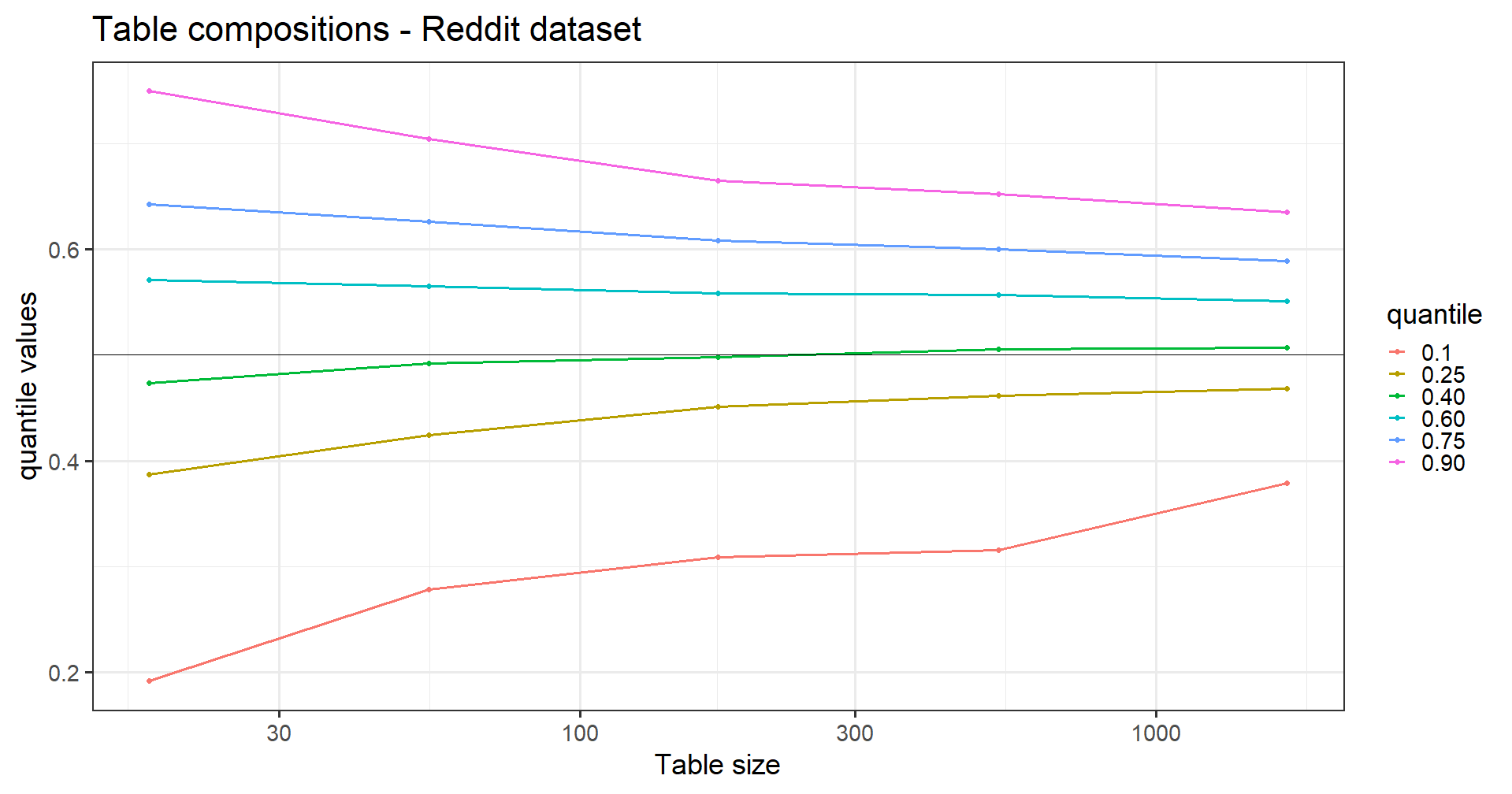}
\caption[]{Reddit data set. 
Real data: Quantiles of the distribution of the proportion of comments with negative sentiment received by the authors  
along their number of received comments, form the least commented to the most commented. } 
\label{fig:Table_composition_Reddit-new}
\end{figure}
  

\subsubsection*{Gutenberg data set}

We downloaded this data set from the on-line library {\em Project Gutenberg}. 
It consists of a collection of over $70\,000$ free ebooks.
After selecting only those written in English and classifying them in different topics, we decided to 
focus on two particular 
{literary genres: ``Western" and ``History". 
For each of them}, we have considered all the words contained in seven books, 
for a total of $480\,460$ words for ``Western" and $476\,948$ words for ``History"  
(after a slight pre-processing: e.g. removal of punctuation, spaces, numbers and words with $1$ or $2$ characters and 
acquisition of the stem of the words 
by means of Dr. Martin Porter's stemming algorithm \cite{porterSA}).  
\\ 

\indent 
{We are interested in the two sequences of words for the two different literary genres and,  
as explained at the beginning of this subsection, we would like to 
check that these two sequences exhibit the asymptotic behaviors predicted by the proposed model. In order to do so, } 
we firstly identify the main quantities related to 
the Chinese restaurant version of the model: each 
{literary genre} is a customer category 
(category 1 = ``Western" and category 2 = ``History") and 
the words represent the tables. 
Therefore, each word will be ``new" or ``old" for a specific 
{literary genre if, respectively, 
it has or has not already been used within that genre. } 
Analogously, each word will be ``new" or ``old" for the entire system if, respectively, 
it has or has not already been used within any considered book.
In order to obtain two sequences of words of the same length,  as required by the model, we have randomly removed some words from the 
category ``Western", i.e. the one containing more words. In addition, we verified successfully that a new word does not appear 
for the first time simultaneously in both genres.\\

\indent 
{For each literary genre $h=1,\,2$, the observed quantity $D^*_{t,h}$ 
(i.e. the number, until time-step $t$, of distinct words 
whose {\em first} appearance has been in literary genre $h$) shows a power law growth along time.  
The same behavior is shown by $D_{t,h}$ (i.e. the number, until-time-step $t$, 
of distinct words used in literary genre $h$). 
Figure \ref{fig:lines_Gutenberg-new} provides the asymptotic behavior of these processes in $\log_{10}-\log_{10}$ scale,  
where we can also appreciate how the lines exhibit the same slope, which indicates that the processes have the same Heaps' exponent.   
This is exactly in accordance with the first result of Theorem~\ref{th-synchro-rates}.  
The estimated value of the common Heaps' exponent, estimated as the common slope of the lines in the $\log_{10}-\log_{10}$ plot is 
$\widehat{\gamma}^*=0.466$. Figure \ref{fig:ratio_Gutenberg} shows the convergence of the process 
$\log_{10}(D^*_{t,1})-\log_{10}(D^*_{t,2})$ toward the estimated limit value $\widehat{u}=-0.238$, 
computed as the difference between the intercepts of the two lines for the two processes $(D^*_{t,h})$
 in Figure \ref{fig:lines_Gutenberg-new}. This value is an estimation of the 
  quantity $u=\log_{10}(r)$, where $r=u_1/u_2=10^u$ is the limit in the second result of Theorem~\ref{th-synchro-rates}. 
  With respect to the Reddit data set, we can observe that here the convergence is slower.}
\\

\begin{figure}[htp]
\centering
\includegraphics[height=0.3\textheight, width=0.8\textwidth]{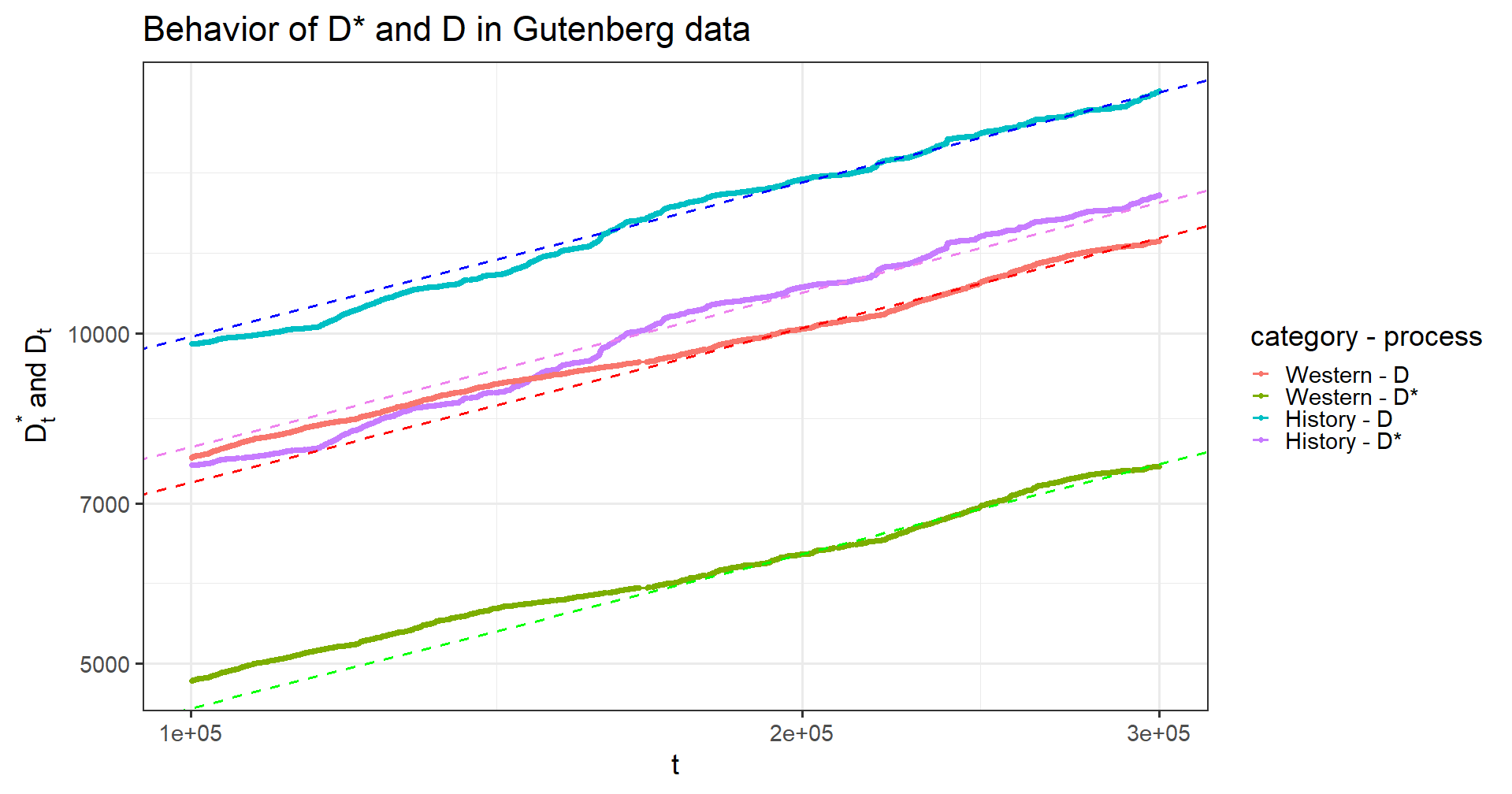} 
\caption[]{
{Gutenberg data set.  
Linear behavior of $(D^{*}_{t,h})$ and $(D_{t,h})$ along time, for $h=1,2$, in $\log_{10}-\log_{10}$ scale. 
The dashed lines are obtained  by a least square interpolation. The goodness of fit $R^2$ index is $0.9937$. 
The estimated common slope is $\widehat{\gamma}^*=0.466$.}}
\label{fig:lines_Gutenberg-new}
\end{figure}

\begin{figure}[htp]
\centering
\includegraphics[height=0.3\textheight, width=0.8\textwidth]{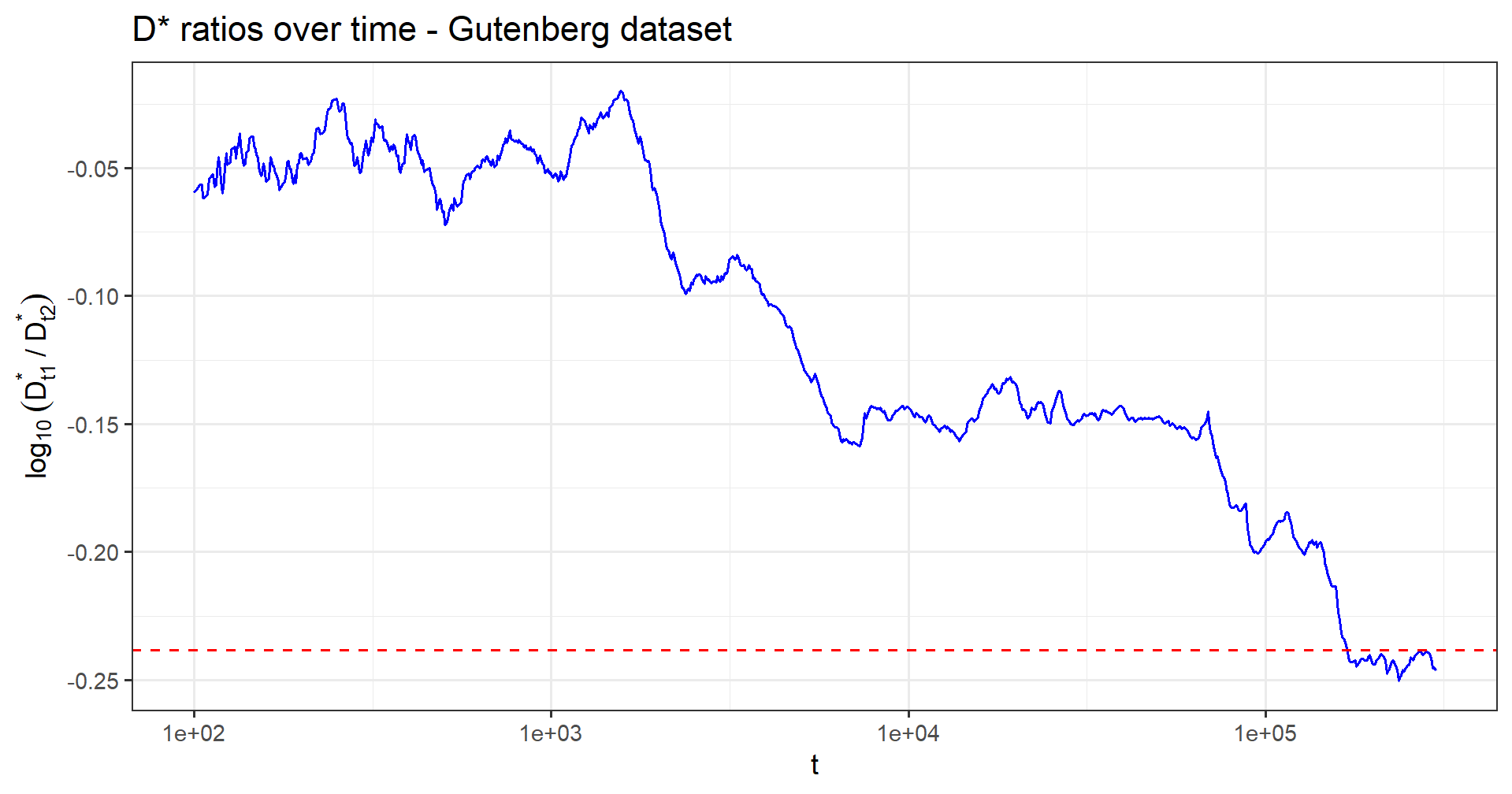}  
\caption[]{
{Gutenberg data set.  
Plot of the process $\log_{10}(D^*_{t,1})-\log_{10}(D^*_{t,2})$ along time. The horizontal dashed red line  
represents the estimated limit value $\widehat{u}=-0.238$.}}
\label{fig:ratio_Gutenberg}.  
\end{figure}

\indent Regarding the table composition, we provide Figure \ref{fig:Table_composition_Gutenberg-new}  with  
the proportion of times a word has been used in the topic ``Western" over the total number of times it 
has been used in the entire system. As in the previous application, we plot the quantiles of the empirical distribution of 
this proportion along the frequency of the words in the system, from the least frequent to the most 
frequent. 
{In order to construct the quantiles, we have listed the words (tables) 
from the least frequent to the most frequent, removing those appeared less than $10$ times (tables with less than $10$ customers),  
then we have grouped these words taking intervals of equal length ($0.5$ in $\log_{10}$ scale). 
Finally, within each group, we have computed the quantiles of the empirical distribution of  
the proportion of times the words have been used in the topic ``Western" with respect to the total number of times it 
has been used in the entire system.  } We can appreciate how these quantiles get closer to $1/2$ 
(the uniform composition) as the frequency of the word increases. This is in accordance with Theorem~\ref{th-uniform-partition}.
\\
 
\begin{figure}[htp]
\centering
\includegraphics[height=0.3\textheight, width=0.8\textwidth]{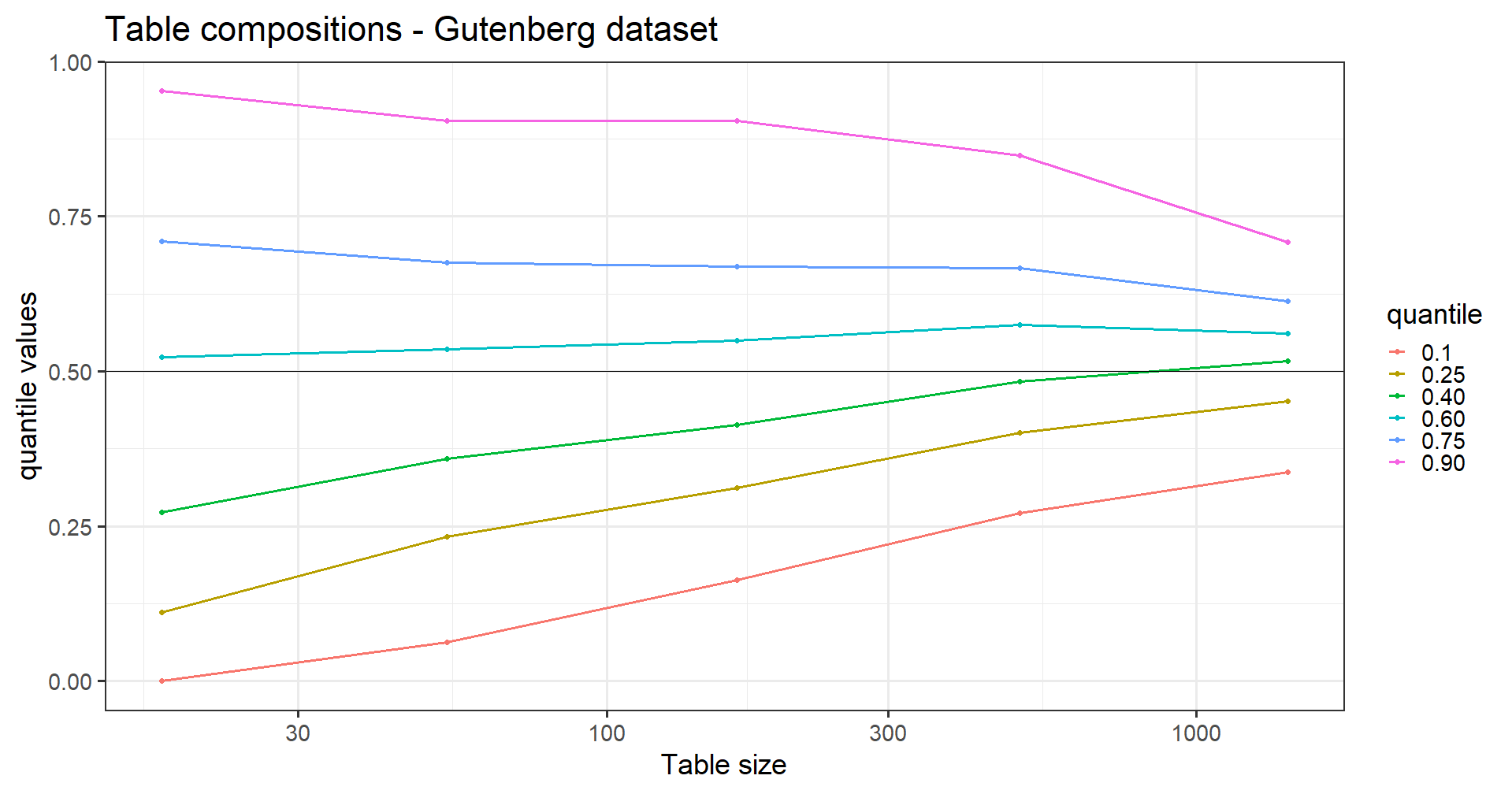}
\caption[]{Gutenberg data set. 
Real data: Quantiles of the distribution of the proportion of times the words have been used in the topic ``Western" 
along their frequency in the system, form the least frequent to the most frequent.  } 
\label{fig:Table_composition_Gutenberg-new}
\end{figure}


\section{Discussion}\label{discussion}

In this work we have introduced a general model in order to analyze a system of $N$ 
{\em interacting} innovation processes. 
The interaction among the processes is ruled by two matrices $\Gamma$ and $W$.  
The first one regulates the production of potential novelties, while the second one tunes the 
interaction with respect to the choice of an old item.  
When matrix $\Gamma$ is irreducible, 
{we have proven 
that the numbers $D^*_{t,h}$, with $h=1,\dots, N$, of distinct novelties for the entire system produced by agent $h$ until time-step $t$ 
have and asymptotic power law behavior with a common Heaps' exponent $0<\gamma^*<1$.  
Moreover, we have proven that the ratio $D^*_{t,h}/D^*_{t,j}$ converges almost surely toward the ratio $u_h/u_j$ 
of the relative centrality scores of $h$ and $j$.
Finally, when the matrix $W$ is irreducible, we have proven that, for each observed item $c$, the 
number of times item $c$ has been adopted by agent $h$ (i.e.~the number of customers of category $h$ sitting at table $c$) 
grows linearly and the proportions of times it has been adopted by agent $h$ over 
the number of times it has been adopted in general in the system converges almost surely to $1/N$ 
 (i.e.~the asymptotic composition of table $c$, with respect to the different $N$ categories, is the uniform one). 
 \\
 \indent In order to highlight the potentialities of the proposed model and of the proven related results 
 in the study of the interaction among innovation processes, we have illustrated that the behaviors predicted by the 
 provided theoretical results match with the ones we observe in two real data sets. 
One interesting research line that we have in mind for the future is to study the speed of convergence for the limits given in 
the shown theoretical results, in order to   
develop statistical instruments for an accurate inference on the two interaction matrices, $\Gamma$ and $W$, from the real data.  
Regarding this issue, it is important to note that the value $\gamma^*$ and the vector $\mathbf{u}=(u_h)_{h=1,\dots,N}$ do not 
uniquely determine the matrix $\Gamma$. In other terms, given the estimates of $\gamma^*$ and of $\mathbf{u}$, there exist 
infinite matrices $\Gamma$ that could have generated that estimated values. 
This is a tough point to deal with and further theoretical results are needed if we want to detect the 
model parameters from the data. 
In the Supplementary material, we present an idea for a first estimation of the interaction matrices in the case $N=2$: 
after the estimation of $\gamma^*$ and $\log_{10}(r)=\log_{10}(u_1/u_2)$ as the common slope and the difference of 
the intercepts, respectively, of the lines related to the observed processes $(D^*_{t,h})$, $h=1,\,2$, plotted in $\log_{10}-\log_{10}$ scale,  
we can consider parametric families of matrices $\Gamma$ and $W$ compatible with these estimated values 
and we can perform a Maximum Likelihood Estimation (MLE) in order to detect the remaining parameters that better fit the data. 
However, for having a robust MLE procedure, we need to reduce the number of parameters by 
imposing some restrictions on them, for instance the symmetry of the matrices. 
We have tested this procedure on some simulations and the results are collected in the Supplementary material.}
\\
\indent Regarding the model assumptions, we point out that the balance condition~\eqref{balance-cond} 
forces to have Heaps' exponents strictly smaller than $1$.   
Since eliminating this condition in the case of a single process ($N=1$) makes an exponent equal to $1$ possible 
\cite{ale-cri_SR, Tria3, Tria1, Tria2},   
it is plausible that it would be the same also for $N\geq 2$. 
Therefore, a second research line for the future is to investigate the proposed model without assuming the balance condition.  Moreover, the 
balance condition forces $w_{h,j}$ (the parameter governing the interaction in the selection of an old item)  
to be large whenever $\gamma_{h,j}$ (the parameter tuning the interaction with respect to the production of potential novelties) 
is large and, vice versa, $\gamma_{h,j}$ is necessarily small whenever $w_{h,j}$ is small.
On the contrary, the proposed model without the restriction of the balance condition may include cases where $\gamma_{h,j}$ is large, 
but $w_{h,j}$ is small.
\\
\indent Another model assumption that could be removed is the simultaneity in the extractions from all the urns of the system 
(i.e. in the arrivals of the customers for all the categories). Indeed, this condition forces to have the same number of observations 
for each process of the system. This variant of the model could be 
obtained by inserting a selection mechanism for the urn from which the extraction at a certain time-step will be performed 
(i.e. for the category of the customer who will enter the restaurant at a certain time-step). This selection could be driven 
by a reinforcement mechanism on the number of times an urn (category) has been selected.
\\
\indent Finally, regarding the assumption of irreducibility in the theoretical results, 
we underline that when the matrices are not irreducible, it is possible 
to decompose them in irreducible sub-matrices such that 
the union of the spectra of the sub-matrices coincides with the spectrum of the original matrix. Then, a deeper analysis starting from the present theory is needed in the same spirit of \cite{ale-ghi, ale-cri-ghi, ale-cri-ghi-complete}. 
{See also the Supplementary material for an heuristic argument in order to deduce the rate at which each $D^*_{t,h}$ grows 
 in the case of a general (i.e.~not necessarily irreducible) matrix $\Gamma$. }
 


\section*{Acknowledgements}
Giacomo Aletti is a member of the Italian Group ``Gruppo
Nazionale per il Calcolo Scientifico'' of the Italian Institute
``Istituto Nazionale di Alta Matematica''. Irene Crimaldi is a 
member of the Italian Group ``Gruppo Nazionale per l'Analisi
Matematica, la Probabilità e le loro Applicazioni'' of the Italian
Institute ``Istituto Nazionale di Alta Matematica''.  
Work partially done while Giacomo Aletti and Andrea Ghiglietti were hosted at
\emph{MATRIX-MFO Tandem Workshop 2023}.
\\[5pt]
\noindent{\bf Funding Sources}\\
\noindent Irene Crimaldi is partially supported by the Italian
``Programma di Attivit\`a Integrata'' (PAI), project ``TOol for
Fighting FakEs'' (TOFFE) funded by IMT School for Advanced Studies
Lucca.

%
%

\appendix

\makeatletter
\let\oldtitle\@title
\let\@title\@empty
\makeatother

\markright{SUPPLEMENTARY MATERIAL}
\newpage
\setcounter{section}{0}
\setcounter{page}{1}
\setcounter{figure}{0}
\setcounter{equation}{0}
\renewcommand{\thesection}{S\arabic{section}}
\renewcommand{\thepage}{s\arabic{page}}
\renewcommand{\thetable}{S\arabic{table}}
\renewcommand{\thefigure}{S\arabic{figure}}
\renewcommand{\theequation}{S:\arabic{figure}}




\section*{\large\textbf{\uppercase{Supplementary material for Interacting Innovation processes: case studies from Reddit and Gutenberg}}}
\thispagestyle{empty}


\section{Analytical proofs}

Denote by $X^*_{t,h}$ the random variable that takes value $1$ when the ball extracted from urn $h$ at time-step $t$ has 
a new (for all the system) color and is equal to $0$ otherwise. Then $Z^*_{t,h}$ defined in \eqref{birth-prob-inter} 
coincides with $P(X^*_{t+1,h}=1\,|\,\mbox{past})=E[X^*_{t+1,j}\,|\,\mbox{past}]$ and 
$D^*_{t,j}$ can be written as $\sum_{n=1}^t X^*_{n,j}$. Therefore, since we have  
 $$
 Z^*_{t,h}=\frac{\theta_h+\sum_{n=1}^t\sum_{j=1}^N \gamma_{j,h}X^*_{n,j}}{\theta_h+t}\,,
 $$
 we obtain the following dynamics for $Z^*_{t,h}$:
 \begin{equation*}
 Z^*_{0,h}=1,\qquad Z^*_{t+1,h}=(1-r_{t,h})Z^*_{t,h}+r_{t,h}\sum_{j=1}^N \gamma_{j,h}X^*_{t+1,j}\quad\mbox{for } t\geq 0\,,
 \end{equation*}
 where $r_{t,h}=1/(\theta_h+t+1)=1/(t+1)+O_h(1/t^2)$. 
 The corresponding vectorial dynamics 
{for $\vZ=(Z_{t,1},\dots,Z_{t,N})^\top$} is 
\begin{equation}\label{eq-dynamics-vector}
  \begin{split}
    {\vZ}^*_0&=\vone\\
    {\vZ}^*_{t+1}&=
    \left(1-\frac{1}{t+1}\right){\vZ}^*_{t}+\frac{1}{t+1} \Gamma^T {\vX}_{t+1}^* + \boldsymbol{O}(1/t^2)\\
    &=
    {\vZ}^*_t - \frac{1}{t+1}(I-\Gamma^\top){\vZ}^*_t + \frac{1}{t+1}\Gamma^\top \Delta {\vM}^*_{t+1} +
    \boldsymbol{{O}}(1/t^2)\quad\mbox{for } t\geq 0,
\end{split}
\end{equation}
where $\Delta{\vM}^*_{t+1}={\vX}^*_{t+1}-{\vZ}^*_t$ and 
$\boldsymbol{{O}}(1/t^2)=(O_1(1/t^2),\dots,O_N(1/t^2))^{\top}$.  
\\

\indent We prove the following key result:

\begin{theorem}\label{th-Zstar}
Under the same assumptions and notation of Theorem~\ref{th-synchro-rates}, we have 
$$
t^{1-\gamma^*}{\vZ}^*_t\stackrel{a.s.}\longrightarrow \widetilde{Z}^{**}_\infty{\vu}\,,
$$
where $\widetilde{Z}^{**}_\infty$ is an integrable strictly positive random variable.
\end{theorem}

\begin{proof} 
We firstly want to decompose the vectorial process ${\vZ}^*_t$ based 
on the Jordan representation of the matrix $\Gamma$. 
Specifically, for any $\gamma\in Sp(\Gamma^\top)\setminus \gamma^*$, 
we can denote as $J_\gamma$ the Jordan block and with $U_\gamma$ and $V_\gamma$ the matrices whose columns are, respectively, the
left and right (possibly generalized) eigenvectors of $\Gamma$ associated to the eigenvalue $\gamma$, i.e. 
$$\Gamma V_\gamma = V_\gamma J_\gamma
\qquad\text{and}\qquad
U_\gamma^\top \Gamma = J_\gamma U_\gamma^\top.$$
Then, we can consider the decomposition
$$
{\vZ}^*_t=\widetilde{Z}^*_t{\vu}+\sum_{\gamma\in Sp(\Gamma^\top)\setminus \gamma^*}{\vZ}^*_{\gamma,t}\,,
$$
where $\widetilde{Z}^*_t={\vv}^\top{\vZ}^*_t$ and
${\vZ}^*_{\gamma,t}= U_\gamma V_\gamma^\top {\vZ}^*$.
Secondly, we set
$$
\zeta_0=1,\qquad \zeta_{t}=1/\prod_{k=1}^{t}\left[1-\frac{(1-\gamma^*)}{k}\right]\sim
t^{1-\gamma^*}\uparrow +\infty 
$$
and 
$$
\vZ^{**}_t= \zeta_t\vZ^{*}_t,
\quad 
\widetilde{Z}^{**}_{t} = \zeta_t\widetilde{Z}^*_{t}\quad\mbox{and}\quad 
{\vZ}^{**}_{\gamma,t} = \zeta_t {\vZ}^*_{\gamma,t}
$$
(note that $\widetilde{Z}^{**}_t$ is non-negative but not bounded by $1$ as $\widetilde{Z}^*_t$)
 so that we have 
$$
\vZ^{**}_t=\widetilde{Z}^{**}_t{\vu}+\sum_{\gamma\in Sp(\Gamma^\top)\setminus \gamma^*}{\vZ}^{**}_{\gamma,t}.
$$ 
In the following steps, we are going to show that 
$\widetilde{Z}^{**}_t$ converges almost surely and in mean to an integrable random variable $\widetilde{Z}^{**}_\infty$
 such that $P(\widetilde{Z}^{**}_\infty>0)=1$ and that each ${\vZ}^{**}_{\gamma,t}$
converges almost surely to zero. In particular, this last task will be done
separately for the eigenvalues with $|\gamma|<\gamma^*$ and with 
$|\gamma|=\gamma^*$. Remember that the assumption that 
$\Gamma$ (or, equivalently, $\Gamma^\top$) is irreducible ensures that 
$\gamma^*$ is real, simple and $|\gamma|\leq \gamma^*$
for any $\gamma\in Sp(\Gamma^\top)$. In the sequel of the proof, the 
symbol $\mathcal{F}_t$ denotes the past until time-step $t$. 
\\
 
 \noindent{\bf Study of $\widetilde{Z}_t^{**}$. } 
By multiplying equation \eqref{eq-dynamics-vector} by $\vv^\top$ we obtain
 \begin{equation*}
 \widetilde{Z}^*_0=1,\qquad 
 \widetilde{Z}^*_{t+1}=\left[1
-\frac{1}{t+1}(1-\gamma^*)\right]\widetilde{Z}^*_{t}+
\frac{1}{t+1}\gamma^*\Delta \widetilde{M}^*_{t+1}+\widetilde{O}\left(\frac{1}{t^2}\right)\, .
 \end{equation*}
 Then, multiplying everything by $\zeta_{t+1}$ and using the relation $\zeta_{t+1}=\zeta_t[1-(1-\gamma^*)/(t+1)]^{-1}$
 we get the following dynamics for $\widetilde{Z}^{**}_{t}=\zeta_{t}\widetilde{Z}^*_{t}$, where 
 $\Delta\widetilde{M}^*_{t+1}=\vv^{\top}\Delta\mathbf{M}^*_{t+1}$, 
 \begin{equation}\label{eq-dynamics-Z_tilde_**}
\begin{split}
\widetilde{Z}^{**}_0=1,\qquad 
\widetilde{Z}^{**}_{t+1}&=\left[1
-\frac{1}{t+1}(1-\gamma^*)\right]\frac{\zeta_{t+1}}{\zeta_{t}}\zeta_{t}\widetilde{Z}^*_{t}+
\frac{\zeta_{t+1}}{t+1}\gamma^*\Delta \widetilde{M}^*_{t+1}+\widetilde{O}\left(\frac{\zeta_{t+1}}{t^2}\right)
\\
&=\widetilde{Z}^{**}_{t}+
\frac{\zeta_{t+1}}{t+1}\gamma^*\Delta \widetilde{M}^*_{t+1}+\widetilde{O}\left(\frac{\zeta_{t+1}}{t^2}\right)\,.
\end{split}
\end{equation}
Therefore, we have  
$$
E[\widetilde{Z}^{**}_{t+1}
|\mathcal{F}_t]=\widetilde{Z}^{**}_t+
\widetilde{O}(\zeta_{t+1}/t^2).
$$ Since $\gamma^*>0$ and so $\sum_t \zeta_{t+1}/t^2\sim
\sum_t 1/t^{1+\gamma^*}<+\infty$, the process $\widetilde{Z}_t^{**}$ is a
non-negative almost (super-)martingale, almost surely convergent 
toward a finite random variable $\widetilde{Z}^{**}_\infty$ (see Appendix \ref{sec-almost-supermart}). 
Then, using Theorem \ref{th-general}, we can prove that $P(\widetilde{Z}^{**}_\infty>0)=1$. Indeed, if 
we define the stochastic process ${\mathcal W}=(\mathcal{W}_t)_{t\geq 0}$, 
taking values in the interval $[0,1]$, as 
\begin{equation}\label{rw-dyn}
\begin{split}
\mathcal{W}_0=\widetilde{Z}_0^*\\
\mathcal{W}_{t+1}&=\left(1-\frac{1}{t+1}\right)\mathcal{W}_r+\frac{1}{t+1} Y_{t+1},\quad t\geq 0,
\end{split}
\end{equation}
where $Y_{t+1}=\gamma^*\widetilde{X}^*_{t+1}$ (that takes values in $[0,1]$, since $\gamma^*<1$,  
$X^*_{t+1}\in\{0,1\}$ and $\mathbf{v}^\top\mathbf{1}=1$), then we have
$$
|\mathcal{W}_{t}-\widetilde{Z}^*_{t}|=O(1/t^2)\to 0
$$
and also
$$
|\zeta_{t}\mathcal{W}_{t}-\widetilde{Z}^{**}_{t}|=|\zeta_{t}\mathcal{W}_{t}-\zeta_{t}\widetilde{Z}^*_{t}|=O(\zeta_{t}/t^2)
=O(1/t^{1+\gamma^*})\to 0\,.
$$
From Theorem \ref{th-general} applied to $(\mathcal{W}_t)$ 
with $\delta=\gamma^*$, we get that $\zeta_{t}\mathcal{W}_t$ converges 
almost surely to a random variable with values in $(0,+\infty)$. 
This random variable is obviously also the almost sure limit of 
$\widetilde{Z}^{**}_t$ and so we can conclude that $P(\widetilde{Z}^{**}_\infty>0)=1$. 
\\
\indent Furthermore, we can observe that, for each $t$, we have  
$|E[\widetilde{Z}^{**}_t]-E[\widetilde{Z}^{**}_0]|\leq 
\sum_{n=0}^{t-1} |E[\widetilde{Z}^{**}_{n+1}] - E[\widetilde{Z}^{**}_n]|\leq
\sum_n |O(\zeta_{n+1}/n^2)|$
 and thus, since the last series is finite,  
 we have $\sup_t E[\widetilde{Z}^{**}_t]<+\infty$. 
 By Fatou's lemma, this fact implies that $\widetilde{Z}^{**}_\infty$ 
 is integrable.
 \\
 \indent Now, we are ready to prove Lemma~\ref{lemma-tec}, whose 
 statement and proof is postponed at the end of the present proof.
 A first consequence of this lemma is that the convergence 
 of $\widetilde{Z}^{**}_t$ to $\widetilde{Z}^{**}_\infty$ is also in mean.
 Indeed,  from \eqref{eq-dynamics-Z_tilde_**}, since 
 $\sup_t E[\widetilde{Z}^{**}_t]<+\infty$ and  
 $(\Delta \widetilde{M}^{*}_{t+1})^2\leq
C \sum_{j=1}^N (\Delta M^{*}_{t+1,j})^2$, we can obtain 
$$
E[(\widetilde{Z}^{**}_{t+1})^2]\leq E[(\widetilde{Z}^{**}_t)^2]+
(\gamma^{*})^2\frac{\zeta_{t+1}^2}{(t+1)^2}CE[V^{*}_t] +
\widetilde{O}(\zeta_{t+1}/t^2)\,,
$$
where $V^{*}_t$ is defined  in the statement of Lemma~\ref{lemma-tec}. 
Then, we find 
\begin{equation*}
\begin{split}
|E[(\widetilde{Z}^{**}_t)^2]-E[(\widetilde{Z}^{**}_0)^2]|&\leq 
\sum_{n=0}^{t-1} |E[(\widetilde{Z}^{**}_{n+1})^2] - E[(\widetilde{Z}^{**}_n)^2]|\\
&\leq
(\gamma^{*})^2\sum_n \frac{\zeta_{n+1}^2}{(n+1)^2}CE[V^{*}_n]+ 
\sum_n |O(\zeta_{n+1}/n^2)|<+\infty\,,
\end{split}
\end{equation*}
where we have used Lemma~\ref{lemma-tec} in order to say that the first series is finite.
Therefore, we have $\sup_{t}E[(\widetilde{Z}^{**}_t)^2]<+\infty$ and so
$(\widetilde{Z}^{**}_t)_t$ is uniformly integrable and we can conclude that 
$\widetilde{Z}^{**}_t$ converges to $\widetilde{Z}^{**}_\infty$ also in mean.
\\

\noindent{\bf Dynamics of ${\vZ}^{**}_{\gamma,t}$.} 
 By multiplying equation~\eqref{eq-dynamics-vector} by $\zeta_{t+1}$ we get
\begin{equation*}
\begin{split}
\zeta_{t+1}{\vZ}^*_{t+1}&=
\zeta_{t+1}{\vZ}^*_t - 
\frac{1}{t+1}\zeta_{t+1}(I-\Gamma^\top){\vZ}^*_t + \frac{1}{t+1}\zeta_{t+1}\Gamma^\top \Delta {\vM}^*_{t+1} +\boldsymbol{{O}}(\zeta_{t+1}/t^2)\\
{\vZ}^{**}_{t+1}&=
\frac{\zeta_{t+1}}{\zeta_{t}}{\vZ}^{**}_t - 
\frac{1}{t+1}\frac{\zeta_{t+1}}{\zeta_{t}}(I-\Gamma^\top){\vZ}^{**}_t + \frac{\zeta_{t+1}}{t+1}\Gamma^\top \Delta {\vM}^{*}_{t+1} +\boldsymbol{{O}}(\zeta_{t+1}/t^2),
\end{split}
\end{equation*}
where ${\vZ}^{**}_{t}=\zeta_{t}{\vZ}^*_{t}$ and
$\Delta {\vM}^{**}_{t}=\zeta_{t}\Delta {\vM}^{*}_{t}$.
Then, using the relation
$\zeta_{t+1}/\zeta_t = 1 + (\zeta_{t+1}/\zeta_t)(1-\gamma^*)/(t+1)$ and recalling that 
$\zeta_{t+1}/\zeta_t=1+O(1/t)$, we obtain
\begin{equation}\label{eq:dynamics_Z**}
\begin{split}
{\vZ}^{**}_{t+1} &= 
{\vZ}^{**}_{t}+\frac{\zeta_{t+1}}{\zeta_{t}}\frac{1-\gamma^*}{t+1}{\vZ}^{**}_{t}-\frac{1}{t+1}\frac{\zeta_{t+1}}{\zeta_{t}}(I-\Gamma^\top){\vZ}^{**}_t + 
\frac{\zeta_{t+1}}{t+1}\Gamma^\top \Delta {\vM}^{*}_{t+1} +\boldsymbol{{O}}(\zeta_{t+1}/t^2)\\
&= 
{\vZ}^{**}_{t}-\frac{1}{t+1}
(\gamma^*I-\Gamma^\top){\vZ}^{**}_t + \frac{\zeta_{t+1}}{t+1}\Gamma^\top \Delta {\vM}^{*}_{t+1} +
\boldsymbol{{O}}(\zeta_{t+1}/t^2).
\end{split}
\end{equation}

\noindent{\bf Study of ${\vZ}^{**}_{\gamma,t}$ with $|\gamma|<\gamma^*$.}
Let $\vB_t=V_\gamma^\top\vZ^{**}$ and since 
${\vZ}^{**}_{\gamma,t}=U_\gamma V_\gamma^\top\vZ^{**}=U_\gamma \vB_t$,
it is enough to prove that $\|\vB_t\|^2$ converges a.s. to zero.
To this end, by multiplying equation \eqref{eq:dynamics_Z**} by $V_\gamma^\top$, we have
$$
{\vB}^{**}_{t+1} = 
\left[I-\frac{1}{t+1}
(\gamma^*I-J_\gamma^\top)\right]{\vB}^{**}_t + \frac{\zeta_{t+1}}{t+1}J_\gamma^\top V_\gamma^\top \Delta {\vM}^{*}_{t+1} 
+\boldsymbol{{O}}(\zeta_{t+1}/t^2).
$$
Then, since for any real matrix $A$ we can write 
\begin{equation}\label{eq:cross_product_Delta_M}
E[\Delta\vM_{t+1}^{*\top}
A
\Delta\vM_{t+1}^*|\mathcal{F}_t] = 
\sum_{j=1}^N a_{jj}^2
E[\Delta M^{*2}_{j,t+1}|\mathcal{F}_t] \leq \max_j{a_{jj}^2} V^{*}_t,
\end{equation}
we have that
$$
\begin{aligned}
E[\|\vB^{**}_{t+1}\|^2|\mathcal{F}_t] &= 
\Big\|\left[ \left(
1 -\frac{\gamma^*}{t+1}\right)I
+ \frac{1}{t+1}J_{\gamma}
\right]\vB^{**}_{t}\Big\|^2 + 
\left(\frac{\zeta_{t+1}^2}{(t+1)^2}\right)
\sum_{j=1}^N [\bar{V}_{\gamma} 
\bar{J}_{\gamma}
J_{\gamma}^\top
V_{\gamma}^\top]_{jj}^2\mathbb{E}[\Delta M^{*2}_{j,t+1}|\mathcal{F}_t]\\
&\leq
\left( 1 -
 \frac{\gamma^*}{t+1}
+ \frac{\|J_{\gamma}\|_{2,2}}{t+1}
\right)^2\|\vB^{**}_{t}\|^2 + 
\left(\frac{\zeta_{t+1}^2}{(t+1)^2}\right)
\max_j\{[\bar{V}_{\gamma} 
\bar{J}_{\gamma}
J_{\gamma}^\top
V_{\gamma}^\top]^2_{jj}\}
V^{*}_t.
\end{aligned}
$$
Then, regarding the first term, we note that
$$
\left( 1 -
 \frac{\gamma^*}{t+1}
+ \frac{\|J_{\gamma}\|_{2,2}}{t+1}
\right)^2 \leq 
\left( 1 -
 \frac{\gamma^*}{t+1}
+ \frac{|\gamma|+\gamma^*}{2(t+1)}
\right)^2 =
\left( 1 -
 \frac{\gamma^* - |\gamma|}{2(t+1)}
\right)^2,
$$
and so
$$
E[\|\vB^{**}_{t+1}\|^2|\mathcal{F}_t] \leq
\left( 1 -
 \frac{\gamma^* - |\gamma|}{2(t+1)}
\right)^2 \|\vB^{**}_{t}\|^2 + 
C\frac{\zeta^2_{t+1}}{(t+1)^2}V^*_t.
$$
Therefore, since $\gamma^*> |\gamma|$ and by Lemma \ref{lemma-tec},
the process $\|\vB^{**}_{t}\|^2$ is a non-negative almost supermartingale 
that converges almost surely.
Moreover, by applying the expectation we obtain
$$
E[\|\vB^{**}_{t+1}\|^2] \leq
\left( 1 -
 \frac{\gamma^* - |\gamma|}{2(t+1)}
\right)^2E[\|\vB^{**}_{t}\|^2] + 
C\frac{\zeta^2_{t+1}}{(t+1)^2}E[V^*_t],
$$
which, since $\sum_t (\gamma^* - |\gamma|)/(t+1)=+\infty$, by 
Lemma \ref{lemma-tec} and Lemma~\ref{lemma-tecP}, we can conclude that 
$\|\vB^{**}_{t}\|\stackrel{a.s.}\longrightarrow 0$,
and hence $\vB^{**}_{t}\stackrel{a.s.}\longrightarrow \vzero$.
\\

\noindent{\bf Study of ${\vZ}^{**}_{\gamma,t}$ with $|\gamma|=\gamma^*$.} 
From the Frobenious-Perron theory, we know that each eigenvalue with maximum modulus is simple.
Then, set $b_t=\vv_\gamma^\top\vZ^{**}$ so that, since we have  
${\vZ}^{**}_{\gamma,t}=\vu_\gamma \vv_\gamma^\top\vZ^{**}=\vu_\gamma b_t$,
it is enough to prove that $|b_t|$ almost surely converges to zero.
To this end, by multiplying equation \eqref{eq:dynamics_Z**} by $\vv_\gamma^\top$, we have
$$
{b}^{**}_{t+1} = 
\left[1-\frac{1}{t+1}
(\gamma^*-\gamma)\right]{b}^{**}_t + \frac{\zeta_{t+1}}{t+1}\gamma \vv_\gamma^\top \Delta {\vM}^{*}_{t+1} +\boldsymbol{{O}}(\zeta_{t+1}/t^2).
$$
Then, using \eqref{eq:cross_product_Delta_M},
we have that
$$
\begin{aligned}
E[|b^{**}_{t+1}|^2|\mathcal{F}_t] &= 
\Big|
1 -\frac{\gamma^*}{t+1}
+ \frac{\gamma}{t+1}
\Big|^2 |b^{**}_{t}|^2 + 
\left(\frac{\zeta_{t+1}^2}{(t+1)^2}\right)
|\gamma|^2\sum_{j=1}^N|v_{j}|^2E[\Delta M^{*2}_{j,t+1}|\mathcal{F}_t]\\
&\leq
\Big|
1 -\frac{\gamma^*}{t+1}
+ \frac{\gamma}{t+1}
\Big|^2 |b^{**}_{t}|^2  + 
\left(\frac{\zeta_{t+1}^2}{(t+1)^2}\right)
|\gamma|^2
\max_j\{|v_j|^2\}
V^{*}_t.
\end{aligned}
$$
Then, regarding the first term we have that
$$
\begin{aligned}
\Big|
1 -\frac{\gamma^*}{t+1}
+ \frac{\gamma}{t+1}
\Big|^2
&=\left(
1 -\frac{\gamma^*}{t+1}
+ \frac{\mathcal{R}e(\gamma)}{t+1}
\right)^2 +
\left( \frac{\mathcal{I}m(\gamma)}{t+1}\right)^2 \\
&=
1 +
 \left( \frac{\gamma^* - \mathcal{R}e(\gamma)}{t+1}\right)^2
  - 2 \left(\frac{\gamma^* - \mathcal{R}e(\gamma)}{t+1}\right)
 + \left( \frac{\mathcal{I}m(\gamma)}{t+1}\right)^2 
\\
&=
1 - \left(\frac{2(\gamma^* - \mathcal{R}e(\gamma))}{t+1}\right)
  +  \left( \frac{\gamma^{*2} 
  -2\gamma^{*}\mathcal{R}e(\gamma) + \mathcal{R}e(\gamma)^2+
  \mathcal{I}m(\gamma)^2}{(t+1)^2}\right)
\\
&=1 - \left(\frac{2(\gamma^* - \mathcal{R}e(\gamma))}{t+1}\right)
  +  \left( \frac{2\gamma^{*} (\gamma^{*} -\mathcal{R}e(\gamma))}{(t+1)^2}\right)
\\
&=1 - 2\left(\frac{1}{t+1} -
  \frac{\gamma^{*}}{(t+1)^2}\right)
  (\gamma^* - \mathcal{R}e(\gamma))
\\
\end{aligned}
$$
and so
$$
E[|b^{**}_{t+1}|^2|\mathcal{F}_t] \leq
\left( 1 - 2\left(\frac{1}{t+1} -
  \frac{\gamma^{*}}{(t+1)^2}\right)
  (\gamma^* - \mathcal{R}e(\gamma))
\right) |b^{**}_{t}|^2 + 
C\frac{\zeta^2_{t+1}}{(t+1)^2}V^*_t.
$$
Therefore, since $\gamma^*> \mathcal{R}e(\gamma)$ and by Lemma \ref{lemma-tec},
the process $|b^{**}_{t}|^2$ is a non-negative almost supermartingale 
that converges almost surely.
Moreover, by applying the expectation, we obtain
$$
E[|b^{**}_{t+1}|^2] \leq
\left(1 - 2\left(\frac{1}{t+1} -
  \frac{\gamma^{*}}{(t+1)^2}\right)
  (\gamma^* - \mathcal{R}e(\gamma))
\right)\mathbb{E}[|b^{**}_{t}|^2] + 
C\frac{\zeta^2_{t+1}}{(t+1)^2}\mathbb{E}[V^*_t]\,.
$$
Since $\sum_t (1/(t+1)-\gamma^*/(t+1)^2)=+\infty$ and by 
Lemma \ref{lemma-tec} and Lemma \ref{lemma-tecP}, we can conclude that 
$|b^{**}_{t}|\stackrel{a.s.}\longrightarrow 0$,
and hence $b^{**}_{t}\stackrel{a.s.}\longrightarrow 0$.

\end{proof}

\begin{lemma}\label{lemma-tec}
Set
$V^*_t=\sum_{j=1}^N E[(\Delta M^*_{t+1,j})^{2}\,|\,\mathcal{F}_t]$.
Then, if $\Gamma$ is irreducible, we have
\begin{equation}\label{lemma-tec-eq}
 \textstyle{\sum_t} \frac{\zeta_{t+1}^2}{(t+1)^2} E[V^*_t] < +\infty\quad\mbox{and so}\quad
 \textstyle{\sum_t} \frac{\zeta_{t+1}^2}{(t+1)^2} V^*_t < +\infty\ \mbox{a.s.} 
\end{equation}
\end{lemma}

\begin{proof} 
First notice that by definition
$$V^*_t=
\sum_{j=1}^N E[(\Delta M^*_{t+1,j})^{2}\,|\,\mathcal{F}_t]=
\sum_{j=1}^N  Z^{*}_{j,t}(1-Z^{*}_{j,t})\leq \sum_{j=1}^N  Z^{*}_{j,t}.$$
Then, denoting by $v_{\min}$ the minimum element of $\vv$, which is strictly positive since $\Gamma^\top$ is irreducible, we have that
$\sum_{j=1}^N  Z^{*}_{j,t}\leq \vv^\top \vZ^{*}_t/v_{\min} = \widetilde{Z}_t^{*}/v_{\min}$. Therefore, we have 
$$\zeta_{t}V^{*}_t\leq \frac{\widetilde{Z}_t^{**}}{v_{\min}}\,.$$ 
Therefore, recalling that $\sup_{t}E[\widetilde{Z}^{**}_t]<+\infty$ and $\frac{\zeta_{t+1}}{(t+1)^2}=O(1/t^{1+\gamma^*})$,  
we get  
$$
E\left[\sum_t \frac{\zeta_{t+1}^2}{(t+1)^2} V^*_t\right]= 
\sum_t \frac{\zeta_{t+1}^2}{(t+1)^2} E[V^*_t] \leq 
\frac{1}{v_{\min}}\sup_{t}E[\widetilde{Z}^{**}_t] \sum_t \frac{\zeta_{t+1}}{(t+1)^2} < +\infty.
$$
This concludes the proof.
 \end{proof}


\subsection*{Proof of Theorem~\ref{th-synchro-rates}}

Leveraging on Theorem~\ref{th-Zstar}, we can prove Theorem~\ref{th-synchro-rates}. Indeed, by the previous convergence results for 
$(Z^*_{t,h})_t$, we have
$$
D^*_{t,h}=\sum_{n=1}^t X^*_{n,h}\qquad\mbox{with } 
E[X^*_{t+1,h}\,|\,\mbox{past}]=Z^*_{t,h}\stackrel{a.s.}\sim \frac{ \widetilde{Z}^{**}_{\infty} u_h }{t^{1-\gamma^*}}
$$
and so, by Lemma~\ref{williams-lemma},  we get 
$$
D^*_{t,h}\stackrel{a.s.}\sim D^{**}_{\infty,h}\,t^{\gamma^*}\quad\mbox{with }
D^{**}_{\infty,h}=\frac{ \widetilde{Z}^{**}_{\infty} u_h }{\gamma^*}\,.
$$
As a consequence, we obtain 
$$
\frac{ D^{*}_{t,h} }{ D^{*}_{t,j} }\stackrel{a.s.}\longrightarrow \frac{ D^{**}_{\infty,h} }{ D^{**}_{\infty,j} }=\frac{u_h}{u_k}\,.
$$


\subsection*{Proof of Theorem~\ref{th-uniform-partition}}

Recall from \eqref{old-color-prob-inter} that, for any color $c$ already present in the network at time $t$, 
$P_t(h,c)=P(C_{t+1,h}=c|\ \mbox{past})$ denotes the conditional probability that
 the extraction at time-step $t+1$ from urn $h$ gives the old color $c$,
while $K_t(h,c)$ indicates the number of times the color $c$ has been drawn 
from urn $h$ until time-step $t$. \\

\indent First of all, we observe that, from \eqref{old-color-prob-inter}, we have  
 $$
 P_{t}(h,c)=\frac{\sum_{j=1}^N w_{j,h}K_t(j,c)-\gamma_{j^*(c),h}}{\theta_h + t}=
 \frac{\sum_{n=1}^t\sum_{j=1}^N w_{j,h}\Delta K_n(j,c)}{\theta_h + t}-\frac{\gamma_{j^*(c),h}}{{\theta_h + t}}
 \,,
 $$
 where $\Delta K_{n}(j,c)=K_{n}(j,c)-K_{n-1}(j,c)$. Notice that
 $\Delta K_{n}(j,c)$ takes values in $\{0,1\}$ and
 $E[\Delta K_{n+1}(j,c)|\mbox{past}]=P_{n}(j,c)$.
Then, we obtain the following dynamics for $P_{t}(h,c)$:
 \begin{equation*}
 P_0(h,c)=0,\qquad 
 P_{t+1}(h,c)=(1-r_{t,h})P_{t}(h,c)+r_{t,h}\sum_{j=1}^N w_{j,h}\Delta K_{t+1}(j,c)\,,
 \end{equation*}
 where $r_{t,h}=1/(\theta_h+t+1)=1/(t+1)+O_h(1/t^2)$.
  Thus the corresponding vectorial dynamics for 
  $\vP_t(c)=(P_t(1,c),\dots, P_t(N,c))^\top$ is
\begin{equation}\label{eq-dynamics-vector-Ptc}
  \begin{split}
\vP_0(c)=\mathbf{0},\qquad    \vP_{t+1}(c)&=
    \left(1-\frac{1}{t+1}\right)\vP_t(c)+\frac{1}{t+1} W^T {\Delta\vK}_{t+1}(c) + 
    \boldsymbol{O}(1/t^2)\\
    &=
    \vP_t(c) - \frac{1}{t+1}(I-W^\top)\vP_t(c) + \frac{1}{t+1}W^\top \Delta {\vM}_{t+1}(c) +
    \boldsymbol{{O}}(1/t^2),
\end{split}
\end{equation}
where  ${\Delta\vK}_t(c)=(\Delta K_t(1,c),\dots, \Delta K_t(N,c))^\top$,
$\Delta{\vM}_{t+1}(c)={\Delta\vK}_{t+1}(c)-\vP_t(c)$ and 
$\boldsymbol{{O}}(1/t^2)=(O_1(1/t^2),\dots,O_N(1/t^2))^{\top}$.   
We can note that the dynamics of $\vP_t(c)$ in \eqref{eq-dynamics-vector-Ptc} 
presents exactly the same form of the dynamics of $\vZ_t^{*}$ in \eqref{eq-dynamics-vector}.
Indeed, the only difference lies in the interacting matrix, which is $W$ in \eqref{eq-dynamics-vector-Ptc}, 
while was $\Gamma$ in \eqref{eq-dynamics-vector}. The different conditions on these two matrices, 
i.e. $W^\top\vone=\vone$ and $\Gamma^\top\vone<\vone$, lead through the 
Frobenious-Perron theory to have different leading eigenvalues, that is  
$w^*=1$ for $W$ and $\gamma^*<1$ for $\Gamma$. Then $\vP_t(c)$ converges almost surely to a strictly positive random variable, 
while, as proven above, $\vZ_t^{*}$ converges almost surely to $\mathbf{0}$. 
To prove the almost sure convergence of $\vP_t(c)$, we can apply exactly the same proof of Theorem~\ref{th-Zstar} 
replacing $\Gamma$ (and the corresponding eigen-structure) by $W$. 
In general this simplifies the proof, e.g. $\zeta_t\equiv 1$ and 
the relation \eqref{lemma-tec-eq} (with $\zeta_t\equiv 1$ and 
$V^*_t=\sum_{j=1}^N E[(\Delta M_{t+1,j}(c))^{2}\,|\,\mathcal{F}_t]$) 
is trivially true. Therefore, since for $W$ we have $\vu=\vone$, we have 
$$
\vP_t(c)\stackrel{a.s.}\longrightarrow 
\widetilde{P}_\infty(c)\vone,
$$
where $\widetilde{P}_\infty(c)$ is a bounded strictly positive random variable. 
The fact that it is strictly positive comes from Theorem \ref{th-general} with $\delta=w^*=1$.  
\\
\indent Finally, since $K_{t}(j,c)=\sum_{n=1}^t\Delta K_{n}(j,c)$
and $E[\Delta K_{n+1}(j,c)|\mbox{past}]=P_{n}(j,c)\stackrel{a.s.}\sim \widetilde{P}_{\infty}(c)$, 
by Lemma~\ref{williams-lemma},  we can conclude that 
$$
K_{t}(h,c) \stackrel{a.s.}\sim \widetilde{P}_{\infty}(c)\, t
$$
and so the statement of Theorem~\ref{th-uniform-partition} holds true with $K_\infty(c)=\widetilde{P}_{\infty}(c)$. 


\subsection{A general result}

Define the stochastic process ${\mathcal W}=(\mathcal{W}_t)_{t\geq 0}$
taking values in the interval $[0,1]$ and following the dynamics
\begin{equation}\label{rw-dyn}
\mathcal{W}_{t+1}=\left(1-\frac{1}{t+1}\right)\mathcal{W}_t+\frac{1}{t+1} Y_{t+1},\quad t\geq 0,
\end{equation}
where $Y_{t+1}$ takes values in $[0,1]$ and is such that 
$E[Y_{t+1}\,|\,\mbox{past}]\stackrel{a.s.}\sim\delta\mathcal{W}_t$ with $0<\delta\leq 1$. 
\\

We are going to prove the following result

\begin{theorem}\label{th-general}
Given $\mathcal{W}_0>0$, we have that $\mathcal{W}_t$ 
converges almost surely to $0$ as $t^{-(1-\delta)}$, that is 
$t^{(1-\delta)}\mathcal{W}_t$ converges almost surely to a random variable with values 
in $(0,+\infty)$.
\end{theorem} 

First of all, we note (see \cite{ale-cri-ghi-barriers} for details)   
that, for each $t$, the random variable $\mathcal{W}_t$ corresponds to the
proportion $H_t/s_t$ of balls of color $A$ inside the urn at time-step $t$ 
for a two-color urn process where the number of balls of color $A$
(resp. $B$) added to the urn at time-step $t$ is $U_t^A=\alpha_{t}Y_t$
(resp. $U_t^B=\alpha_t(1-Y_t)$) with $\alpha_{t} =
\frac{1/t}{\prod_{n=1}^{t} (1-1/n)}\sim 1$ (and so $s_t=1/\prod_{n=1}^t(1-1/k)\sim t$).
Note that, if $(\mathcal{F}_t)_t$ is the filtration
associated to the urn process, we have  
\begin{equation}\label{cond-mean}
E[U_{t+1}^A\vert \mathcal{F}_t]\stackrel{a.s.}\sim 
\alpha_{t+1}\delta\mathcal{W}_t\,.
\end{equation}
We observe also that, since $Y_t$ takes values in $[0,1]$ and so $Y^2_t\leq Y_t$,  
we have 
 \begin{equation}\label{cond-second-moment}
E[(U_{t+1}^A)^2\vert \mathcal{F}_t]\leq 
\alpha_{t+1}^2 E[Y_{t+1}\,|\,\mathcal{F}_t]
\stackrel{a.s.}\sim \delta \mathcal{W}_t\,.
\end{equation}
In the following two lemmas we will show that $H_t$ diverges almost surely to $+\infty$ and $1/H_t=o(t^{-1/\theta})$ for $\theta>1/\delta$. 

\begin{lemma}\label{lemma-infinito}
Assuming $\mathcal{W}_0>0$,  $H_t$ diverges almost surely to $+\infty$
\end{lemma}

\begin{proof}
Since $H_t=\mathcal{W}_0+\sum_{n=1}^t U^A_n$, where 
the random variables $U_n^A$ are positive and uniformly bounded by a constant. 
By Lemma~\ref{del-mey-result-app}, we have 
$H_t\stackrel{a.s.}\longrightarrow +\infty$ if and only if $\sum_t
E[U_{t+1}^A\vert \mathcal{F}_{t}]=+\infty$ almost surely. Therefore, it is
enough to observe that this last condition is satisfied when $\mathcal{W}_0>0$, 
because of \eqref{cond-mean} and the fact that  $\mathcal{W}_t\geq \mathcal{W}_0\frac{1}{s_t}\stackrel{a.s.}\sim 
 \mathcal{W}_0/t$.
 \end{proof}
 
 \begin{lemma}\label{lemma-o}
 For each $\theta>1/\delta>1$, we have $1/H_t=o(t^{-1/\theta})$.
 \end{lemma}

\begin{proof} We have  
  \begin{equation*}
    \begin{split}
      &E\left[\frac{t+1}{H_{t+1}^\theta}-
        \frac{t}{H_t^\theta}\,\vert \,\mathcal{F}_t\right]=
      E\left[\frac{t+1}{H_t^\theta}-\frac{t}{H_t^\theta} +
        \frac{t+1}{H_{t+1}^\theta}-\frac{t+1}{H_t^\theta}\,\vert \,\mathcal{F}_t
        \right]=\\
      &\frac{1}{H_t^\theta}+
      E\left[(t+1)
        \left(\frac{1}{(H_{t}+U_{t+1}^A)^\theta}-\frac{1}{H_t^\theta}
        \right)\,\vert \,\mathcal{F}_t\right]
        \leq\\
        &\frac{1}{H_t^\theta}+t E\left[
        \left(\frac{1}{(H_{t}+U_{t+1}^A)^\theta}-\frac{1}{H_t^\theta}
        \right)\,\vert \,\mathcal{F}_t\right]\,.
      \end{split}
  \end{equation*}
  Let $C$ so that $0\leq U_{t}^A=\alpha_{t+1}Y_t\leq C$.  
  Using the Taylor expansion of the function
  $f(x)=1/(a+x)^\theta$ (that is
  $f(x)-f(0)=f'(0)x+\frac{f''(x_0)}{2}x^2$ with $x_0\in (0,x)$) 
  with
  $a=H_t$ and $x=U_{t+1}^A$, we have eventually (so that $H_t\geq 1$)
\begin{equation*}
  \begin{split}
  \frac{1}{(H_t+U^A_{t+1})^{\theta}}&-\frac{1}{H_t^\theta}\\
  &\leq
  -\frac{\theta}{H_t^{\theta+1}}U_{t+1}^A+
  \frac{\theta(\theta+1)}{H_t^{\theta+2}}(U_{t+1}^A)^2
  \leq
  -\frac{\theta}{H_t^{\theta+1}}U_{t+1}^A+
  \frac{\theta(\theta+1)}{H_t^{\theta+2}}C U_{t+1}^A
  \end{split}
\end{equation*}
  and so, recalling that $\mathcal{W}_t=H_t/s_t\stackrel{a.s.}\sim H_t/t$, we get 
  \begin{equation*}
    \begin{split}
 E\left[\frac{1}{(H_t+U^A_{t+1})^{\theta}}-\frac{1}{H_t^\theta}\vert \mathcal{F}_t\right]
 &\leq   
  -\frac{\theta}{H_t^{\theta+1}}
  \alpha_{t+1}E[Y_{t+1}\,|\,\mathcal{F}_t]
  \left(1+
  \frac{(\theta+1)C}{H_t}\right)
  \\
  &\stackrel{a.s.}\sim -\frac{\theta\delta}{H_t^{\theta}}
  \frac{1}{t}\left[1+
  O\left(\frac{1}{H_t}\right)\right]\,.
  \end{split}
\end{equation*}
Therefore, we have  
$$
E\left[\frac{t+1}{H_{t+1}^\theta}-\frac{t}{H_t^\theta}\,\vert \,\mathcal{F}_t\right]
\leq
\frac{1}{H_t^\theta t}
\left[-(\theta\delta-1)+O\left(\frac{1}{H_t}\right)\right]
$$ and so, for $\theta\delta>1$, since $H_t\to +\infty$, we can conclude that the above conditional
expectation is eventually negative. This proves that, for each $\theta>1/\delta$, 
 $(t/H_t^{\theta})_t$ is eventually a (positive) super-martingales and
 so, for each $\theta>1/\delta$, it converges almost surely to a finite random variable. 
 Since $\theta>1/\delta$ is arbitrary,
 we necessarily have that $t/H_t^{\theta}$ converges almost surely to zero. 
 This fact concludes the proof.
\end{proof}
 
Now we are ready for the proof of the previous theorem.

\begin{proof} (of Theorem~\ref{th-general})\\ 
Set $L_t=\ln(H_t/t^{\gamma^*})$, $\Delta_t=E[L_{t+1}-L_t\vert \mathcal{F}_t]$ and
$Q_t=E[(L_{t+1}-L_t)^2\vert \mathcal{F}_t]$. If we prove that
 $\sum_t\Delta_t$ and $\sum_t Q_t$ are almost surely convergent, then
$L_t$ converges almost surely to a finite random variable (see Lemma~\ref{lemma-pemantle}). 
This fact implies that $H_t/t^{\gamma^*}$ 
converges to a random variable with values in $(0,+\infty)$. The rest of the proof is devoted to verify that
$\sum_t\vert \Delta_t\vert <+\infty$ and $\sum_t Q_t<+\infty$ almost
surely.\\ To this regard, we note that
  \begin{equation*}
  \begin{split}
    \Delta_t=
    &E[\ln(H_{t+1})-\ln(H_t)\vert \mathcal{F}_t]-\gamma^*\left(\ln(t+1)-\ln(t)\right)=
    \\
    &E[\ln(H_{t}+U_{t+1}^A)-\ln(H_t)\vert \mathcal{F}_t]-\gamma^*\ln(1+1/t)=
    \\
   &E\left[\int_0^{U_{t+1}^A}\frac{1}{H_t+x}\,dx\right]-\gamma^*\ln(1+1/t)\,.
  \end{split}
  \end{equation*}
  Since $1/(H_t+x)\leq 1/H_t$ and $\ln(1+1/t)\geq 1/t-1/(2t^2)$
  for each $x\geq 0$ and each $t$,    the last term of the above equalities
  is smaller than or equal to 
  $$
  \frac{1}{H_t}E[U_{t+1}^A\vert \mathcal{F}_t]
  -\frac{\gamma^*}{t}
  +\frac{\gamma^*}{2t^2}
$$ and so, recalling \eqref{cond-mean} and that $\mathcal{W}_t=H_t/s_t\stackrel{a.s.}\sim H_t/t$, 
it  is smaller than or equal to
$$
\frac{\alpha_{t+1}E[Y_{t+1}|\mathcal{F}_t]}{H_t}-\frac{\delta}{t}+\frac{\delta}{2 t^2}
\stackrel{a.s.}\sim
\frac{\delta}{t}-\frac{\delta}{t}+\frac{\delta}{2 t^2}=O(1/t^2)\,.
$$ 
Therefore $\Delta_t=O(1/t^2)$. Finally, we
note that $-\Delta_t=\delta\ln(1+1/t)-\ln(H_{t+1})+\ln(H_t)$. 
Using $\ln(1+1/t)\leq 1/t$  
 and $1/(H_t + x) \geq 1/H_t - x/H_t^2$ for each $x\geq 0$ and each $t$, 
 we find that
$-\Delta_t$  is smaller than or equal to 
 $$
 \frac{\delta}{t}-\frac{1}{H_t}E[U_{t+1}^A\vert \mathcal{F}_t]+
 \frac{1}{2H_t^2}E[(U_{t+1}^A)^2\vert \mathcal{F}_t]
 $$
 and so, recalling \eqref{cond-mean}, \eqref{cond-second-moment} and 
 that $\mathcal{W}_t=H_t/s_t\stackrel{a.s.}\sim H_t/t$, 
it  is smaller than or equal to
$$
 \frac{\delta}{t}-\frac{\alpha_{t+1}E[Y_{t+1}|\mathcal{F}_t]}{H_t}+
 \frac{\alpha_{t+1}^2E[Y_{t+1}|\mathcal{F}_t]}{2H_t^2}
 \stackrel{a.s.}\sim
 \frac{\delta}{2tH_t}=O(1/(tH_t)).
 $$
By the previous Lemma, we have $1/H_t=o(t^{-\eta})$ for some $\eta>0$ and so 
$-\Delta_t=O(1/t^{1+\eta})$.  Thus,
$\sum_t\vert \Delta_t\vert <+\infty$ almost surely.  
Similarly we have
 \begin{equation*}
  \begin{split}
&E[( \ln(H_{t+1})-\ln(H_t)-\delta\ln(t+1)+\delta\ln(t) )^2
\vert \mathcal{F}_t]
\leq\\
&2\left\{
E[(\ln(H_{t+1})-\ln(H_t))^2\vert \mathcal{F}_t]+
\delta(\ln(t+1)-\ln(t))^2
\right\}
\leq\\
&2 E\left[\left(\int_0^{U_{t+1}^A}\frac{1}{H_t+x}\,dx\right)^2\Big\vert \mathcal{F}_t\right]+
2(\delta)^2/t^2
\leq \\
&2
E[(U_{t+1}^A/H_t)^2\vert \mathcal{F}_t]+
O(1/t^2)
\leq \frac{1}{H_t^2}\alpha_{t+1}^2E[Y_{t+1}|\mathcal{F}_t]
\stackrel{a.s.}\sim \\
&O(1/(tH_t))+O(1/t^2)\,.
  \end{split}
 \end{equation*}
 Therefore, we get $Q_t=O(1/t^{1+\eta})$ for some $\eta>0$ and so
 $\sum_tQ_t<+\infty$ almost surely.
\end{proof}

%
%
%
%

\subsection{Non-negative almost supermartingale}\label{sec-almost-supermart}

Let $(Y_n)$ be an $\mathcal{F}$-adapted sequence of non-negative random variables satisfying
$$
E[Y_{n+1}|\mathcal{F}_n]\leq (1+\Delta_n)Y_n+R_{1,n}-R_{2,n},\,
$$
where $\Delta_n$, $R_{1,n}$, $R_{2,n}$ are all non-negative sequences of random variables.
Then $(Y_n)$ is called non-negative almost super-martingale. 
\\ 

\indent By \cite{rob}, we know that it almost surely converges on $\{\sum_n\Delta_n<+\infty\,,\sum_n R_{1,n}<+\infty\}$.


\subsection{Some technical results}

For the reader's convenience, 
we here recall some technical results used in the previous proofs.

\begin{lemma}[{\cite[Supplementary material]{ale-cri-ghi-complete}}]\label{lemma-tecP}
If $a_t\geq 0$, $a_t \leq 1$ for $t$ large enough, $\sum_t a_t =
+\infty$, $\delta_t \geq 0$, $\sum_t \delta_t < +\infty$, $b>0$,
$y_t\geq 0$ and $y_{t+1} \leq (1-a_t)^b y_t + \delta_t$, then $\lim_t
y_t = 0$.
  \end{lemma}

\begin{lemma}[{\cite[Theorem~46, p.~40]{del-mey}}]\label{del-mey-result-app}
  Let $(Y_t)_t$ be a sequence of non-negative random variables, adapted to a
  filtration ${\mathcal F}=(\mathcal{F}_t)_t$. Then the set $\{\sum_t
  E[Y_{t+1}\vert \mathcal{F}_{t}]<+\infty\}$ is almost surely contained in
  the set $\{\sum_t Y_t<+\infty\}$. If the random variables $Y_t$ are
  uniformly bounded by a constant, then these two sets are almost
  surely equal.
\end{lemma}

\begin{lemma}[{\cite[Sec.~12.15]{williams}}]\label{williams-lemma}
Let $(Y_t)_t$ be a sequence of Bernoulli random variables, adapted to a filtration $\mathcal{F}=(\mathcal{F}_t)_t$ 
and  such that $Z_t=P(Y_{t+1}=1\,|\,\mathcal{F}_t)$. 
Then $\sum_{n=1}^t Y_n / \sum_{n=0}^{t-1} Z_n \stackrel{a.s.}\longrightarrow 1$.
\end{lemma}

\begin{lemma}[{\cite[Lemma~3.2 ]{pemantle-volkov-1999}}]\label{lemma-pemantle} 
  Let $(L_n)_n$ be a sequence of random variables, adapted to a
  filtration $\mathcal{G}_n$. Set
  $\Delta_n=E[L_{n+1}-L_n\vert \mathcal{G}_n]$ and
  $Q_n=E[(L_{n+1}-L_n)^2\vert \mathcal{G}_n]$. If $\sum_n\Delta_n$ and
  $\sum_nQ_n$ are almost surely convergent, then $(L_n)_n$ converges
  almost surely to a finite random variable.
\end{lemma}


\section{Heuristics}\label{heuristics}

{
We here describe an heuristic argument (also employed in \cite{iacopini-2020}), useful in order to detect the rate at which 
 each $D^*_{t,h}$ grows along time in the case of a general matrix $\Gamma$.}\\ 
 
{\indent The dynamics that rules the vectorial process $\mathbf{D}^*_t=(D^*_{t,1},\dots,D^*_{t,N})^\top$
 can be approximated (as $t\to +\infty$) by the linear system of (deterministic) differential equations 
\begin{equation*}
\dot{\mathbf{d}}^*(t) = \Gamma \frac{\mathbf{d}^*(t)}{t}
\end{equation*}
and hence we can say that $\mathbf{D}^*_t\approx \mathbf{d}^*_t$ for $t\to +\infty$. 
By the change of variable $t=e^z$, we get
\begin{equation*}
\dot{\mathbf{d}}^*(z) = \Gamma \mathbf{d}^*(z)\,,
\end{equation*}
whose general solution is given by $\mathbf{d}^*(z)=e^{\Gamma z}\mathbf{c}$. Now,
the term $e^{\Gamma z}$ can be expressed using the canonical Jordan form of the matrix $\Gamma$, 
so that we obtain
$$
\mathbf{d}^*(z)=\sum_{k=1}^{r} e^{\gamma_k z} \sum_{i=0}^{p_k-1} z^i \mathbf{c}_i,
$$
where  $\gamma_1,\dots, \gamma_r$ are the distinct eigenvalues of $\Gamma$, $p_1,\dots,p_r$ are 
the sizes of the corresponding Jordan blocks and $\mathbf{c}_i$ are suitable vectors related to $\mathbf{c}$ and to 
the generalized eigenvectors of $\Gamma$. Indeed, 
we can write $\Gamma$ as $PJP^{-1}$, where $J$ is its canonical Jordan form and $P$ is a suitable invertible matrix of 
generalized eigenvectors.
Therefore, we have $e^{\Gamma z}=P e^{J z} P^{-1}$, where $e^{Jz}$ is a block matrix with blocks of the form
$e^{J_k z}$ with $J_k$ block in $J$. On the other hand, if $J_k =\gamma_k I + N_k$ is a generic Jordan 
block of $\Gamma$ with size $p_k$ and associated to the eigenvalue $\gamma_k$, we have
$$
e^{J_k z}= e^{\gamma_k z} e^{N_k z}=
e^{\gamma_k z} \sum_{i=0}^{p_k-1} \frac{z^i}{(i-1)!} N_k^i\, .
$$
Changing the variable from $z$ to $t$, we find 
\begin{equation}\label{eq-general-solution}
\mathbf{D}^*_t\approx \mathbf{d}^*(t) = 
\sum_{k=1}^{r} t^{\gamma_k} \sum_{i=0}^{p_k-1} \ln^i(t) \mathbf{c}_i
\end{equation}
and so the rate at which $D^*_{t,h}$ increases is given by the leading term in the expression
of $d^*_h(t)$. }
\\

{In particular, when $\Gamma$ is irreducible,
the above general formula leads, for each $D^*_{t,h}$, to the same asymptotic behavior 
$t^{\gamma^*}$, with $\gamma^*$ equal to the leading eigenvalue of $\Gamma$ 
(recall that $\gamma^*$ is simple and so the logarithm term is not present). However, it is important to 
note that, with this heuristic argument, we can deduce the right rate at which each 
$D^*_{t,h}$ grows, but we cannot get any information about the limit random variable: 
we can deduce that, for each $h$,  the quantity $D^*_{t,h}/(u_h t^{\gamma^*})$, 
where $\mathbf{u}$ is the vector of the relative centrality scores,  converges almost surely 
to a certain random variable (first statement of Theorem~\ref{th-synchro-rates}), but we cannot affirm 
that these limit random variables are all equal and this last fact is fundamental in order to obtain 
the second statement of Theorem~\ref{th-synchro-rates}. 
Nevertheless, we can affirm that the merit of this heuristics is the fact that, 
from \eqref{eq-general-solution}, we can get the 
rate at which each $D^*_{t,h}$ grows for any matrix $\Gamma$.}

\section{A preliminary idea for the estimation of the interaction in the case $N=2$}
\label{sec-estimate-interaction}
 
{In this section, for the case $N=2$, we provide a parametric family for the matrix 
$\Gamma=(\gamma_{j,h})_{j,h=1,2}$ such that its leading eigenvalue $\gamma^*$ and the ratio $r=u_1/u_2$ of the components of 
its corresponding left eigenvector coincide with some given values. More precisely, given 
the values $\gamma^*\in (0,1)$ and $r\in (0,1]$, the matrices 
\begin{equation}\label{eq:formula_gamma_real_data}
\Gamma(x_1,x_2) =
\begin{pmatrix}
\gamma^*(1-x_1) & \frac{\gamma^*}{r}x_2I_{(\gamma^*\leq r)}+\frac{(1-\gamma^*)}{(1-r)}x_2I_{(\gamma^*>r)}\\
r\gamma^*x_1    & \gamma^*(1-x_2)I_{(\gamma^*\leq r)}+\left[\gamma^*-\frac{(1-\gamma^*)}{(1-r)}rx_2\right]I_{(\gamma^*>r)}
\end{pmatrix},\qquad x_1,\,x_2\in (0,1)
\end{equation}
are non-negative, irreducible, such that $\vone^\top\Gamma<\vone^\top$ and have 
the leading eigenvalue equal to $\gamma^*$ and the ratio of the components of the 
corresponding left eigenvector equal to $r$.  Moreover, we can define a parametric family for the matrix $W=(w_{j,h})_{j,h=1,2}$, 
adding other two parameters,  as  $ W(x_1,x_2,y_1,y_2)=\Gamma(x_1,x_2)+\Lambda(x_1,x_2,y_1,y_2)$ where
 \begin{equation*}\label{eq:formula_W}
\Lambda(x_1,x_2,y_1,y_2)=\begin{pmatrix}
(1-[\Gamma(x_1,x_2)^{\top} \vone]_1)(1-y_1) & (1-[\Gamma(x_1,x_2)^{\top} \vone]_2)y_2\\
(1-[\Gamma(x_1,x_2)^{\top} \vone]_1)y_1 & (1-[\Gamma(x_1,x_2)^{\top} \vone]_2)(1-y_2)
\end{pmatrix} ,\qquad y_1,\,y_2\in [0,1].
\end{equation*}
Note that the above matrices $W(x_1,x_2,y_1,y_2)$ are non-negative, irreducible and such that 
$\vone^\top W=\vone^\top$. The balance condition is satisfied by construction. 
}
\\

{Given a data set such that the observed processes exhibit asymptotic behaviors in accordance with the 
provided theoretical results of the model, the above parametric families for the two interaction matrices $\Gamma$ and $W$ 
can be used for performing a Maximum Likelihood Estimation (MLE) procedure. In details: 
\begin{itemize}
\item[1)] estimate the quantity $\gamma^*$ as the common slope of the lines in the $\log_{10}-\log_{10}$ plot of the processes 
$(D^*_{t,h})$, with $h=1,\,2$;
\item[2)] estimate the quantity $r$ as $10^{\widehat{u}}$, where $\widehat{u}$ is the difference between the intercepts of 
   the lines in the $\log_{10}-\log_{10}$ plot of the processes $(D^*_{t,h})$, with $h=1,\,2$ 
   (note that, in order to employ the above parametric 
   families of matrices,  we need to label the two categories so that the estimated value for $r$ is $\leq 1$, i.e.~$\widehat{u}\leq 0$); 
\item[3)] consider the matrices $\Gamma(x_1,x_2)$ and $W(x_1,x_2,y_1,y_2)$ related to the estimated values for $\gamma^*$ and $r$;
\item[4)] perform a MLE procedure in order to estimate from the data the interaction parameters 
$x_1,\,x_2,\, y_1$ and $y_2$ and, possibly, the initial parameters $\theta_1$ and $\theta_2$. 
\end{itemize}
However, in order to get a robust MLE estimation, we may want to reduce the number of parameters by imposing some conditions on them:
for instance, we can take $\theta_1$ and $\theta_2$ equal to some given values and  
restrict to matrices $\Gamma(x_1,x_2)$ and $W(x_1,x_2,y_1,y_2)$ that are symmetric 
(which means that the interaction mechanism is symmetric, 
i.e.~the influence of $h=1$ on $h=2$ is equal to the one of $h=2$ on $h=1$). 
The general formula of the likelihood function that we have to maximize is:
\begin{multline*}
\mathcal{L}(\theta_1,\theta_2,x_1,x_2,y_1,y_2;c_{1,1},c_{1,2},\dots,c_{T,1},c_{2,T})=
\\
\prod_{t=1}^{T-1}\prod_{h=1}^2 
\left( Z^*_{t,h} I_{\{c_{t+1,h}\, \text{is new}\}}+ P_t(h,c) I_{\{c_{t+1,h}\,\text{is equal to an old item } c\}} \right)
\end{multline*}
where $I_E$ denotes the indicator function of the event $E$, 
$Z^*_{t,h}$ and $P_t(h,c)$ are given in~\eqref{birth-prob-inter} and 
in~\eqref{old-color-prob-inter} , respectively, and  
$(c_{t,1})_{1,\dots,T}$ and $(c_{t,2})_{1,\dots,T}$ are the 
two observed sequences of items (colors/tables) for the two agents (urns/categories) $h=1,\,2$. 
\\}

{We now present a simulation study aimed at highlighting the performance of the 
estimation procedure obtained by following the steps 1)-4) of the algorithm proposed above. 
In order to reduce the number of parameters to be estimated, we set $\theta_1=\theta_2=1$ and 
we impose that both $\Gamma$ and $W$ must be symmetric. 
This assumption, combined with the condition $W^{\top}\vone=\vone$, implies that $\Gamma$ and $W$ can be univocally identified by four 
parameters, e.g. $\gamma_{1,1}$, $\gamma_{1,2}$, $\gamma_{2,2}$, $w_{1,2}$. For each choice of $\Gamma$ and $W$, 
$100$ independent innovation processes following the model presented in this work have been generated until the time-step $T=10^4$. 
Then, we have applied steps 1)-4) to the data generated by each simulation, so obtaining a set of $100$ estimates of $\gamma^{*}$, $r$, 
$x_1$, $x_2$, $y_1$ and $y_2$ which fulfill the symmetric condition, 
i.e. each one leading to symmetric estimated matrices $\widehat{\Gamma}$ and $\widehat{W}$.
The results of this simulation study are collected in Table~\ref{tab:table-simulations}, 
where the mean values and the standard deviations of the estimated elements are compared with the true ones used for 
generating the data. Regarding the elements of the two matrices and $\gamma^*$, 
the estimation procedure works very well in all the cases.
Regarding $r$, we can note that the estimated values are "sensitive" to the strenght of the interaction term $\gamma_{1,2}$: 
the higher the interaction term, the better is the estimation.}

\begin{table}[h!]
{
  \begin{center}
    \caption{Simulation results of the estimation procedure described in steps 1)-4) with
    $\theta_1=\theta_2=1$ and assuming $\Gamma$ and $W$ symmetric. 
    Each parameter has been estimated by $100$ independent simulated processes generated until time-step $T=10^4$.\\
Columns 1-4: elements of the interacting matrices $\Gamma$ and $W$ used to generate the data.\\
Columns 5-8: mean values and standard deviations of the elements of the 
$100$ estimated interacting matrices $\widehat{\Gamma}$ and $\widehat{W}$.\\
Colmuns 9-10: true $\gamma^*$ and $r$ .\\
Columns 11-12: mean values and standard deviations of the $100$ estimates of $\widehat{\gamma^*}$ and $\widehat{r}$.}
    \label{tab:table-simulations}
\begin{footnotesize}
    \begin{tabular}{|cccc|cccc||cc|cc|}
    \hline
 $\gamma_{1,1}$ & 
 $\gamma_{2,2}$ & 
 $\gamma_{1,2}$ & 
 $w_{1,2}$ &
 $\widehat{\gamma}_{1,1}$ & $\widehat{\gamma}_{2,2}$ & 
 $\widehat{\gamma}_{1,2}$ & $\widehat{w}_{1,2}$ &
 $\gamma^{*}$ & $r$ & 
 $\widehat{\gamma}^{*}$ &  $\widehat{r}$ \\
      \hline
      0.10 & 
      0.40 & 
      0.10 & 
      0.50 & 
      0.10 (0.07) & 0.38 (0.05) &
      0.14 (0.03) & 0.52 (0.03) &
      0.43 & 
      0.30 & 0.43 (0.04) & 0.42 (0.11) \\
      0.10 & 
      0.40 & 
      0.10 & 
      0.25 & 
      0.14 (0.08) & 0.39 (0.05) &
      0.12 (0.03) & 0.26 (0.02) &
      0.43 & 
      0.30 & 0.44 (0.05) & 0.41 (0.10) \\
      0.25 & 
      0.40 & 
      0.10 & 
      0.50 & 
      0.27 (0.06) & 0.39 (0.05) &
      0.12 (0.03) & 0.50 (0.03) &
      0.45 & 
      0.50 & 0.46 (0.04) & 0.63 (0.16) \\
      0.25 & 
      0.40 &
      0.10 & 
      0.25 & 
      0.28 (0.06) & 0.39 (0.05) &
      0.11 (0.03) & 0.25 (0.02) &
      0.45 & 
      0.50 & 0.46 (0.04) & 0.63 (0.14) \\
      0.10 & 
      0.40 & 
      0.25 & 
      0.50 & 
      0.10 (0.05) & 0.39 (0.04) &
      0.26 (0.03) & 0.51 (0.03) &
      0.54 & 
      0.57 & 0.54 (0.03) & 0.59 (0.06) \\
      0.10 & 
      0.40 & 
      0.25 & 
      0.25 & 
      0.13 (0.05) & 0.40 (0.04) &
      0.25 (0.01) & 0.25 (0.01) &
      0.54 & 
      0.57 & 0.54 (0.03) & 0.60 (0.05) \\
      0.25 & 
      0.40 & 
      0.25 & 
      0.50 & 
      0.26 (0.04) & 0.40 (0.04) &
      0.25 (0.03) & 0.50 (0.03) &
      0.59 & 
      0.74 & 0.59 (0.03) & 0.76 (0.07) \\
      0.25 & 
      0.40 & 
      0.25 & 
      0.25 & 
      0.27 (0.04) & 0.40 (0.04) &
      0.24 (0.02) & 0.25 (0.02) &
      0.59 & 
      0.74 & 0.59 (0.02) & 0.76 (0.08) \\
      0.10 & 
      0.40 & 
      0.40 & 
      0.50 & 
      0.11 (0.04) & 0.40 (0.04) &
      0.40 (0.03) & 0.50 (0.03) &
      0.68 & 
      0.69 & 0.68 (0.02) & 0.69 (0.03) \\
      0.25 & 
      0.40 & 
      0.40 & 
      0.50 & 
      0.25 (0.02) & 0.40 (0.03) &
      0.40 (0.02) & 0.50 (0.02) &
      0.73 & 
      0.83 & 0.73 (0.02) & 0.83 (0.03) \\
            \hline
    \end{tabular}
\end{footnotesize}
  \end{center}
  }
\end{table}

{In order to complete the picture, we have also checked how the results can be affected by the choice of 
$\theta_h$ and, in particular, if choosing a wrong value of $\theta_h$ in the likelihood could considerably worsen the estimation of 
$\Gamma$ and $W$. To this end, we have considered some of the scenarios presented in Table~\ref{tab:table-simulations} 
and we have computed the estimates of the elements of $\Gamma$ and $W$ for two different values of 
$\theta_h$ and, in particular, including the cases when the value of $\theta_h$ used to generate the simulated data sets is 
different from the value of $\theta_h$ used to compute the likelihood. 
The results of this simulation study on the "sensitivity" of the parameter $\theta_h$ are collected in Table~\ref{tab:table-simulations_theta}.
In general, we can notice that the results seem to be quite robust to the choice of $\theta_h$ used in the likelihood.
Therefore, the problem of using the "right" $\theta_h$ in the likelihood does not seem so important as we could imagine.
However, the performance of the estimation procedure does worsen considerably when the data are generated with high values of 
$\theta_h$. This is probably due to the fact that, when $\theta_h$ is large, the asymptotic behaviors of the innovation processes 
are reached after a number of time-steps which is much larger than $T=10^4$ used in this simulation study.
\\
\indent In conclusion, the estimation procedure provided in this subsection is only a first step toward the 
estimation of the interaction between two innovation processes. Additional simulations and analyses are needed. In particular, 
we need to understand how to test the restrictions on the parameters, for example how to provide a test on the symmetry of the 
interaction mechanism.}

\begin{table}[h!]
{
  \begin{center}
    \caption{Simulation results of the estimation procedure described in steps 1)-4) with
    $\theta_1=\theta_2$ and assuming $\Gamma$ and $W$ symmetric.  
    Each parameter has been estimated by $100$ independent simulated processes generated until time-step $T=10^4$.\\
Columns 1: value of $\theta_1=\theta_2=\theta_{Data}$ used to generate the data.\\
Columns 2: value of $\theta_1=\theta_2=\theta_{Likelihood}$ put in the likelihood function.\\
Columns 3-6: elements of the interacting matrices $\Gamma$ and $W$ used to generate the data.\\
Columns 7-10: mean values and standard deviations of the elements of the 
$100$ estimated interacting matrices $\widehat{\Gamma}$ and $\widehat{W}$.}
    \label{tab:table-simulations_theta}
    \begin{tabular}{|cc|cccc|cccc|}
    \hline
 $\theta_{Data}$ &  $\theta_{Likelihood}$ & 
 $\gamma_{1,1}$ & 
 $\gamma_{2,2}$ & 
 $\gamma_{1,2}$ & 
 $w_{1,2}$ &
 $\widehat{\gamma}_{1,1}$ & $\widehat{\gamma}_{2,2}$ & 
 $\widehat{\gamma}_{1,2}$ & $\widehat{w}_{1,2}$ \\
      \hline
            1 & 1 &
      0.10 & 
      0.40 & 
      0.10 & 
      0.50 & 
      0.10 (0.07) & 0.38 (0.05) &
      0.14 (0.03) & 0.52 (0.03) \\
      1 & 100 &
      0.10 & 
      0.40 & 
      0.10 & 
      0.50 & 
      0.09 (0.07) & 0.38 (0.05) &
      0.13 (0.03) & 0.51 (0.03) \\
      100 & 1 &
      0.10 & 
      0.40 & 
      0.10 & 
      0.50 & 
      0.22 (0.02) & 0.41 (0.02) &
      0.18 (0.01) & 0.5 (0.01)\\
      100 & 100 &
      0.10 & 
      0.40 & 
      0.10 & 
      0.50 & 
      0.19 (0.02) & 0.40 (0.02) &
      0.19 (0.01) & 0.51 (0.01)\\
            \hline
      1 & 1 &
      0.10 & 
      0.40 & 
      0.10 & 
      0.25 & 
      0.14 (0.08) & 0.39 (0.05) &
      0.12 (0.03) & 0.26 (0.02) \\
      1 & 100 &
      0.10 & 
      0.40 & 
      0.10 & 
      0.25 & 
      0.11 (0.09) & 0.38 (0.05) &
      0.13 (0.03) & 0.26 (0.02) \\
      100 & 1 &
      0.10 & 
      0.40 & 
      0.10 & 
      0.25 & 
      0.28 (0.02) & 0.43 (0.01) &
      0.14 (0.01) & 0.26 (0.01) \\
      100 & 100 &
      0.10 & 
      0.40 & 
      0.10 & 
      0.25 & 
      0.27 (0.02) & 0.43 (0.01) &
      0.15 (0.01) & 0.26 (0.01) \\
      \hline
      1 & 1 &
      0.25 & 
      0.40 & 
      0.40 & 
      0.50 & 
      0.25 (0.02) & 0.40 (0.03) &
      0.40 (0.02) & 0.50 (0.02) \\
      1 & 100 &
      0.25 & 
      0.40 & 
      0.40 & 
      0.50 & 
      0.25 (0.03) & 0.40 (0.03) &
      0.40 (0.02) & 0.50 (0.02) \\
      100 & 1 &
      0.25 & 
      0.40 & 
      0.40 & 
      0.50 & 
      0.29 (0.02) & 0.42 (0.02) &
      0.41 (0.01) & 0.50 (0.01) \\
      100 & 100 &
      0.25 & 
      0.40 & 
      0.40 & 
      0.50 & 
      0.29 (0.02) & 0.42 (0.02) &
      0.41 (0.01) & 0.50 (0.01) \\
             \hline
    \end{tabular}
  \end{center}
  }
\end{table}



\end{document}